\documentclass[12pt]{article}

\usepackage{setspace}
\usepackage{amsmath}
\usepackage{times}
\usepackage{bm}
\usepackage{natbib}
\usepackage{fullpage}
\usepackage{comment}
\usepackage[plain,noend]{algorithm2e}
\usepackage{algorithmic}
\usepackage{amssymb}
\usepackage{mathrsfs}
\usepackage{placeins}
\usepackage{bbm}
\usepackage{amsthm}
\usepackage[normalem]{ulem}
\usepackage{graphicx} 
\usepackage{comment}
\usepackage{xcolor}
\usepackage{booktabs}
\usepackage{float} 
\usepackage{caption}
\usepackage{subcaption}
\usepackage[hidelinks]{hyperref}

\usepackage{cleveref}

\usepackage{mathtools,xparse}
\usepackage{amsfonts}
\usepackage{amsmath}

\def\RR{{\mathbb R}}

\def\PP{{ \mathsf{P}}}

\def\totspace{\Xi}

\newcommand{\A}{\mathcal{A}}
\newcommand{\B}{\mathcal{B}}

\newcommand{\G}{\mathcal{G}}

\renewcommand{\L}{\mathcal{L}}

\newcommand{\iid}{\stackrel{\rm iid}{\sim}} 
\newcommand{\ind}{\stackrel{\rm ind}{\sim}}

\newcommand{\OO}{\mathcal{O}}
\renewcommand{\RR}{\mathbb{R}}

\newcommand{\indicator}{\mathbbm{1}}

\def\dist{\mathrm{dist}}

\newcommand{\bv}{\boldsymbol{v}}

\newcommand{\bxi}{\boldsymbol{\xi}}

\newcommand\data{\mathrm{data}}

\newcommand*{\medcap}{\mathbin{\scalebox{1.4}{\ensuremath{\cap}}}}%
\newcommand*{\medcup}{\mathbin{\scalebox{1.4}{\ensuremath{\cup}}}}%

\crefname{assumption}{Assumption}{Assumptions}
\crefname{proposition}{Proposition}{Propositions}
\crefname{theorem}{Theorem}{Theorems}
\crefname{section}{Section}{Section}
\crefname{lemma}{Lemma}{Lemma}
\crefname{algorithm}{Algorithm}{Algorithms}
\crefname{example}{Example}{Examples}
\crefname{figure}{Figure}{Figure}
\crefname{condition}{Condition}{Condition}
\crefname{appendix}{Appendix}{Appendix}
\crefname{equation}{equation}{equation}
\crefname{table}{Table}{Table}

\newtheorem{theorem}{Theorem}

\newtheorem{proposition}[theorem]{Proposition}

\newtheorem{condition}{Condition}[section]
\newtheorem{definition}[theorem]{Definition}

\newcommand{\lap}{\mathrm{Laplace}}

\newcommand{\Nor}{\mathrm{Normal}}

\DeclareMathOperator*{\argmax}{\arg\!\max}
 
\def\CC{C\texttt{++}}
\def\pidx{l}
\newcommand{\ntotal}{L}  

\usepackage{amsthm}

\usepackage{xr}

\makeatletter
\newcommand*{\addFileDependency}[1]{
\typeout{(#1)}
\@addtofilelist{#1}
\IfFileExists{#1}{}{\typeout{No file #1.}}
}
\makeatother

\title{Bayesian inferences on uncertain ranks and orderings: Application to ranking players and lineups}

\author{Andr\'es F. Barrientos$^1$,  Deborshee Sen$^2$, Garritt L. Page$^{3,4}$, and David B. Dunson$^5$}
\date{%
$^1$Department of Statistics, Florida State University\\%
$^2$Department of Mathematical Sciences, University of Bath\\%
$^3$Department of Statistics, Brigham Young University\\%
$^4$BCAM - Basque Center for Applied Mathematics\\
$^5$Department of Statistical Science, Duke University
}

\begin{document}

\maketitle

\begin{abstract}

\noindent
It is common to be interested in rankings or order relationships among entities. In complex settings where one does not directly measure a univariate statistic upon which to base ranks, such inferences typically rely on statistical models having entity-specific parameters. These can be treated as random effects in hierarchical models characterizing variation among the entities. In this paper, we are particularly interested in the problem of ranking basketball players in terms of their contribution to team performance. Using data from the National Basketball Association (NBA) in the United States, we find that many players have similar latent ability levels, making any single estimated ranking highly misleading.  The current literature fails to provide summaries of order relationships that adequately account for uncertainty.  Motivated by this, we propose a Bayesian strategy for characterizing uncertainty in inferences on order relationships among players and lineups. Our approach adapts to scenarios in which uncertainty in ordering is high by producing more conservative results that improve interpretability. This is achieved through a reward function within a decision theoretic framework. We apply our approach to data from the 2009-10 NBA season.\\

\noindent 
\textbf{Keywords} ~ 
Bayesian; 
Ordering statements;  
Ranking;  
Decision theory;  
Sports statistics.
\end{abstract}



\section{Introduction}

Making inference on orderings among parameters is of widespread importance. These may represent, for example, abilities of individuals or teams performing a given task \citep{cattelan2013dynamic}, treatment effects in clinical trials \citep{rucker2015ranking}, health and social indices of geographical regions \citep{marriott2017caste}, consumer preferences in search engines \citep{park2015search}, and educational systems \citep{millot2015international}. Standard statistical techniques for inferring such orderings aim to estimate a ranking of parameters. These have classically included ranking the posterior means or using the posterior expected ranks \citep{laird1989empirical, lin2006loss}; we refer the reader to \cite{henderson2016making} for a more nuanced technique. Other methods have been developed based on the type of data at hand.  For example, when data originate directly from paired comparisons between competing units, the Bradley-Terry model (\citealt{bradley1952rank}) has become a popular option to estimate each unit's ``ability'' which are then used to produce a ranking (\citealt{agresti, hunter:2004,caron2012efficient}).  Recent work dedicated to Bradley-Terry type models has focused on generalizations of the base model.  For example,  \cite{dittrich1998modelling}  allow unit-specific covariates to inform ability estimates, and \cite{huang2008ranking} use group information to estimate individual abilities.  When preference data are available, \cite{vitelli2018probabilistic} recently developed a method based on the Mallows model that produces rankings.  In this same setting, \cite{Caron:2014} developed a method that produces rankings based on a Plackett-Luce model.

This article is motivated by our interest in providing an approach for producing ordering statements for entities in professional sports leagues, with a particular focus on the National Basketball Association (NBA) of the United States of America. In what follows, we use the term ``ordering statement'' to mean the collection of rankings of a subset of objects of interest. We are interested in comparing players not on the basis of their individual statistics but in terms of their contribution to team performance.  Our primary dataset is from the 2009-10 NBA season, and our goal is to provide ordering statements associated with individuals' and groups of individuals' (lineups') abilities. This is a particularly challenging case study because  many players have similar abilities and, based on the available data, the abilities need to be estimated from measurements taken at a group level, making them imprecise.   Ranking statements typically summarize orderings in parameters and it can be challenging to assess and summarize uncertainty in such statements. This is especially true in scenarios where parameter values are similar  or the amount of data is limited relative to the number of parameters (as in our case study).  In such settings, it is particularly important to not over-interpret a single estimated ranking but instead to carefully account for uncertainty in rank-based decision making and inferences.

Despite the rich statistical literature on ranking problems, few contributions focus on characterizing uncertainty. In the Bayesian framework, methods for measuring uncertainty include point-wise credible intervals for ranks \citep{rodriguez2015measuring} and identifying rankings with high posterior probability (\citealp{soliman2009ranking, vitelli2018probabilistic}; see also \citealp{jewett2019optimal} for some visual tools). In the frequentist framework, some theoretical contributions focus on studying asymptotic conditions under which point-wise confidence intervals of population ranks have the claimed coverage probability \citep{xie2009confidence}, and under which ranking estimates converge to the truth as the number of parameters and sample size increase to infinity \citep{hall2010modeling}. Other contributions aim to define confidence regions for rankings using multiple confidence intervals or hypothesis tests for the parameters while controlling the family-wise error rate \citep{wright2018primer, klein2018simple, mohamadetal:2021}. However, controlling the family-wise error rate can be challenging and often leads to regions that are so wide as to be practically meaningless. Loss functions have also been used in this regard, with the choice of the loss depending on the ranking problem at hand; examples include ranking by using a weighted loss that improves estimates of extrema \citep{wright2003loss}, finding the correct rank for a particular parameter, finding a parameter corresponding to a particular rank, or correctly ordering a pair of parameters (all in \citealp{jewett2019optimal}), identifying the top or bottom rankings \citep{lin2006loss, henderson2016making}, optimizing the area under a receiver operating characteristic curve \citep{rendle2009bpr}, and learning to rank \citep{chen2009ranking}.

In this paper, we adopt a Bayesian framework and propose an approach for producing ordering statements of parameters along with their associated uncertainty. These ordering statements are not designed to make definitive conclusions about all parameters being compared; instead, they maintain accuracy by remaining silent when uncertainty associated with the ranking is high. These ordering statements thus represent uncertainty more realistically as compared to usual rankings. We achieve this by forming two particular types of ordering statements, which we term local and global statements. Local statements provide orderings relative to individual parameters, and global statements combine local statements.  In our framework, combining local statements to form global statements leads to a potentially massive number of global statements. To address this, we propose a decision-theoretic approach to pick a global statement that balances the number of elements included in the statement and its uncertainty. The proposed approach can be parallelised, which makes it attractive from a computational perspective. We also show that the optimal global statement that maximizes the expected utility function converges to the true ranking. We highlight the fact that our approach is based solely on posterior distribution probabilities.  Thus,  it can be employed regardless of the statistical model used to fit the data.  For this reason, we focus primarily on developing and presenting our approach that produces local and global ordering statements and ancillarily on the model employed to fit the NBA data.

The rest of the paper is organised as follows. \cref{sec:challenges} discusses the challenges of ranking parameters under different scenarios; this is motivated concretely by the NBA application. Because our approach is applicable regardless of the statistical model utilized,  \cref{sec:contributions} first provides formal definitions of local and global statements, introduces the reward function used to find optimal statements, and presents an asymptotic result showing the convergence of the optimal statements to the true ranking.
\cref{sec:data.model} describes the statistical model used in the NBA application, and the rest of \cref{sec:simulation_study} is dedicated to conducting a simulation study to assess the performance of our approach and compare to the cumulative probability consensus ranking described in \cite{vitelli2018probabilistic}. \cref{sec.application} presents results for our analysis of the NBA data.  Computer codes used to carry out the analysis in \cref{sec.application} can be found in our R-package at \url{https://github.com/anfebar/anfebar.github.io/tree/master/Software/BOARS}.
Finally, we end with some concluding remarks in \cref{sec:discussion}.


\section{Motivation through basketball application}\label{sec:challenges}

Since our focus is on NBA performance data, we begin by briefly providing some context. In basketball, a competition between two teams is based on five players within a lineup. A basketball game is a timed collection of competitions between lineups, which we refer to as ``encounters''. The lineup compositions for both teams change frequently during the course of a game based on substitutions made by team coaches. It is common for teams to record a number of metrics during encounters to evaluate lineup performance. Among these are the total number of points scored and the total number of points conceded. In this work, we use the difference of these two totals as a way to measure the ability of a lineup while taking into account the quality of opponent. The available information we have is thus at the lineup level rather than the player level.  

It is often stated that the NBA is a ``player's league'', meaning that a team's success is highly dependent on its top players' talent level.  Being able to discriminate between player abilities and attract high-performing players is thus crucial for general managers/decision makers. In addition, basketball is a sport that requires precise interactions between players, and thus being able to identify collections of players that play effectively together is a coach's key responsibility. This makes it important to be able to order players and lineups. Examples of ordering statements that could prove useful to coaches and general managers are a ranking of all players or lineups across the NBA,  rankings of subsets of players or lineups (within a team or between teams),  pairwise comparisons between players and/or lineups, lists of players or lineups that are at the top or bottom of a team or a subset of teams, et cetera. To obtain such statements, a common approach is to first fit a model that produces a posterior distribution of player/lineup abilities, and then compute a ranking based on either a summary measure of the posterior distribution or by means of loss functions. Examples of summary measures include  posterior means, the so-called $r$-values \citep{henderson2016making}, and posterior expected ranks \citep{laird1989empirical, lin2006loss}, while loss functions are explored in \cite{jewett2019optimal}.

If a posterior summary measure is used to construct an ordering of players/lineups, it is desirable that this be accompanied by a measure of uncertainty.  Since an ordering statement defines a region of the parameter space, a natural measure of uncertainty is the amount of posterior probability that such a region accumulates. Ideally, one would like the posterior probability of an ordering statement to be large and also for it to be precise. For example, it would be undesirable to report a wide credible interval for the position at which a player is ranked even though the posterior probability of the interval is high. \cref{fig.caterpillar} illustrates the challenges of producing precise ordering statements with small uncertainty in our application. This displays a caterpillar plot of each players' average per game difference between points scored and allowed for each encounter in which the player participates. This could be thought of as a type of adjusted plus/minus \citep{sill:2010} in that it takes into account lineup information. It is clear that there is a massive overlap in the average difference in points scored and allowed across players. Reasons for this are: \emph{(i)} player abilities are inferred using measurements collected at an aggregate, that is, lineup level, \emph{(ii)} players have very similar abilities, which is not unexpected for a top-level competition, \emph{(iii)} the number of observations for a player varies widely from around twenty to two thousand, and \emph{(iv)} encounters are often competitive, resulting in differences in points scored close to zero. This overlap makes it so that any ordering statements (as those listed in the previous paragraph) that contain a moderate to large subset of players (and thus lineups) would likely have negligible posterior probability. For this reason, it is imperative to attach a measure of uncertainty to these types of ordering statements so that coaches and general managers are aware that they are making decisions based on statements that may be very uncertain.

\begin{figure}
\centering
\includegraphics[ width=0.7\textwidth, trim={0 0cm 0cm 2cm},clip=true]{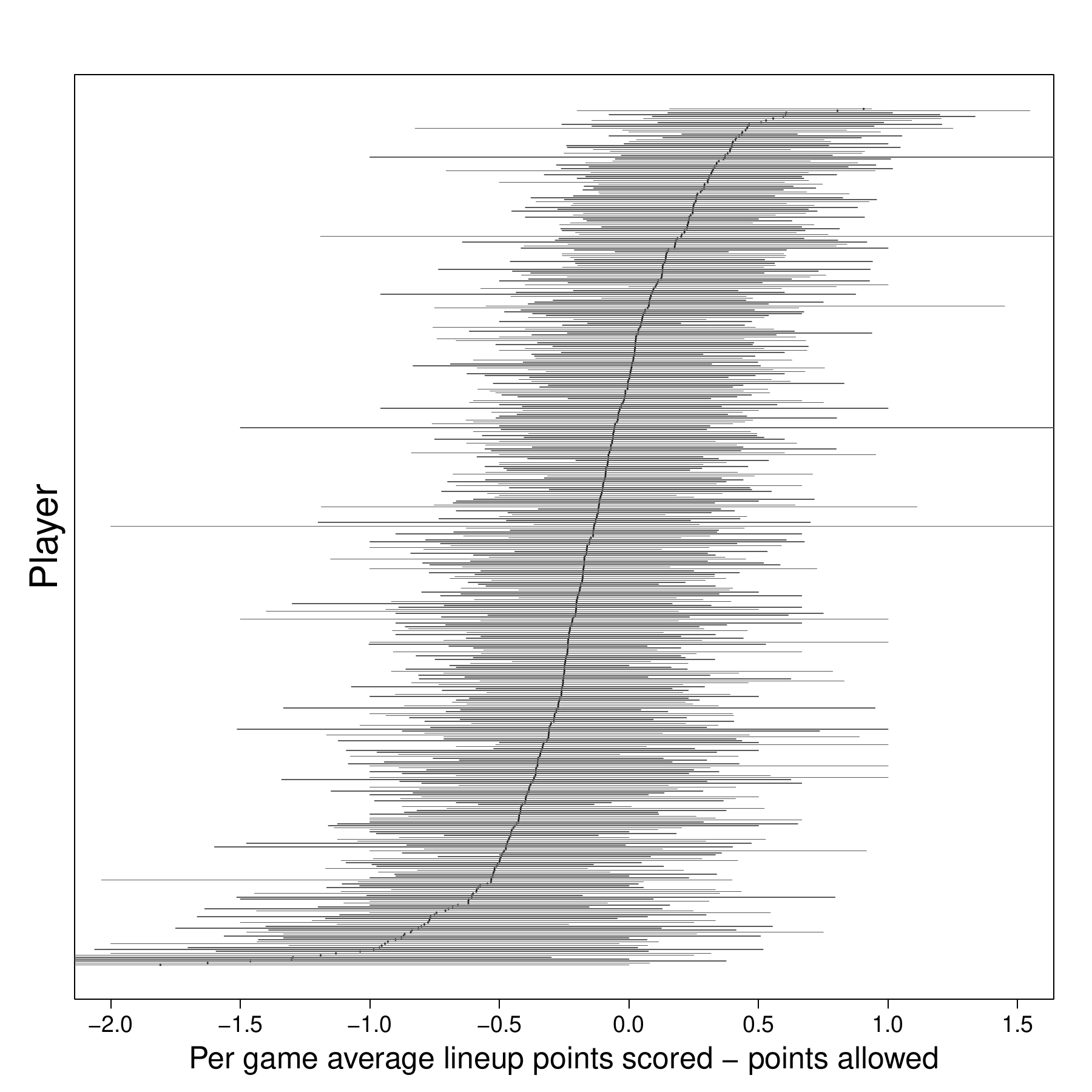}
\caption{Per game average of points scored minus points allowed of each encounter in which a player participated. Each bar represents the the interquartile range of the per game average between points scored and allowed of the all encounters in which a player participated.}
\label{fig.caterpillar}
\end{figure}

In scenarios like our NBA application where any ranking/ordering of a large number of players/lineups is likely to have massive uncertainty, an alternative option is to focus on making statements for a small subset of players/lineups. This can include, for example, pairwise comparisons. Doing so would produce ordering statements that are likely to be less uncertain. However, if a coach/general manager makes decisions based on a collection of statements that consist of a small subset of players/lineups, then in order to correctly quantify uncertainty, they will need to compute the joint posterior probability of each of the ``smaller'' statements occurring simultaneously. Unfortunately, this joint probability is again likely to be low. In \cref{sec:results_players}, we provide some examples of how typical approaches lead to ordering statements which have posterior probability that are essentially zero.

In this article, we propose an approach that attempts to find ordering statements which accumulate a large joint posterior probability. Our approach focuses on exploring the support of the posterior distribution and identifying a region (or regions) in the parameter space for which a non-negligible concentration of posterior probability occurs, and in turn connecting the region to a set of ordering statements. The posterior probability of the identified regions is equal to the joint probability of the corresponding ordering statements, which corresponds to our measure of uncertainty. In addition, not only do we attempt to find an ordering statement on which a non-negligible amount of posterior probability is concentrated, but, if possible,  we find one that involves a reasonable number of players/lineups as well. We achieve this by forming local and global statements, where local statements provide orderings relative to individual parameters, and global statements combine local statements; concrete definitions of these statements are provided in \cref{sec:contributions}. When comparing the abilities of two players/lineups within a local statement, say $\xi_1$ and $\xi_2$, the global statement should refrain from saying whether player 1 is better or worse than player 2 (that is, $\xi_1 > \xi_2$ or $\xi_1 < \xi_2$, respectively) if the posterior probability of the event $\xi_1 > \xi_2$ is not reasonably large or small, respectively. Thus, local statements can end up being empty. When the global statement only involves a few local statements comprised of a few players/lineups, we call it a sparse global statement. If sparsity needs to be reduced, there are two options: \emph{(i)} work with a smaller posterior probability, or \emph{(ii)} introduce a mechanism that allows coaches or managers to introduce some error and accept that some players or lineups may be misordered in some of the ordering statements. Thus, for each ordering statement, we need to provide not only the posterior probability, but also a proper quantification of such error. Our proposed approach to achieve this is both novel and complementary to the existing literature of alternate methods to estimate ordering statements.


Decision makers will be able to employ the ordering statements our approach provides with confidence that uncertainty is explicitly and properly calibrated.  However, it may be the case that comparisons between two or more players/lineups of interest is not available in the global statement, implying that uncertainty associated with the comparisons is large. If this is the case, then the coach/general manager's options are twofold:  \emph{(i)} make a decision based on, for example, the posterior mean of the abilities, $r$-values, posterior expected ranks in the face of high uncertainty (which essentially boils down to a ``guess''), or \emph{(ii)} collect a different set of data that is able to better discriminate between desired players or lineups and run our procedure based on these data. The next section presents a detailed description of our proposed approach. 

Before describing our approach, we briefly mention that analyzing player contributions to lineup effectiveness and overall team strength has appeared in the basketball analytics literature.  A small sampling of these are  \cite{kalman2020nba}, who discover effective lineups based on new position definitions, \cite{page:2007}, use box-score variables to highlight player skills that are conducive to winning games, and \cite{deshpande2016estimating}, who use situational contributions to estimate player quality and produce a player ranking. That being said, we are not aware of any work in sports analytics that addresses uncertainty quantification associated with rankings or orderings.

\section{Bayesian  ordering statements} \label{sec:contributions} 

This section contains the main conceptual contribution of this article. We begin by providing a formal description of the setup we are working under in \cref{sec.setup}. We then introduce notions of local and global statements in \cref{sec.ranking_statements}. These are formed using  loss parameters that users are free to pick according to their needs; this is discussed further in \cref{sec.ranking_statements}. We provide a decision-theoretic approach to tune the loss parameters
in \cref{sec.decision}; this leads to optimal global statements. We show in \cref{sec.asynptotic} that the optimal global statement converges to the true ranking asymptotically as the sample size increases.

\subsection{Setup} \label{sec.setup}

Formally, our goal is to compare $\ntotal$ unknown parameters $\xi_1, \dots, \xi_{\ntotal} \in \totspace$; these represent player abilities in our application, and will also later be used to construct lineup abilities. We use the shorthand notation $\bxi = (\xi_1, \dots, \xi_{\ntotal}) \in \totspace^\ntotal$, and we let $\B(\Xi^L)$ denote the Borel sigma-algebra on $\Xi^L$. We assume that $\totspace^\ntotal$ allows orderings among parameters, that is, for every $(\xi_1, \dots, \xi_{\ntotal}) \in \totspace^\ntotal$ and every pair $(\pidx,\pidx')$, $1 \leq \pidx \neq \pidx' \leq \ntotal$, one of $\xi_{\pidx'} > \xi_\pidx$ or $\xi_\pidx = \xi_{\pidx'}$ or $\xi_\pidx > \xi_{\pidx'}$ is true.

We adopt a Bayesian approach in this article and assume that there is a statistical model with parameters $\bxi$, and we let $\PP$ be the posterior distribution of $\bxi$ 
given observations, where we omit its explicit dependence on observations to keep the notation simple; $\PP$ is defined on the space $(\Xi^L, \B(\Xi^L))$.
It is reasonable to assume that no two players have exactly the same ability. In other words, either $\xi_{\pidx'} > \xi_\pidx$ or $\xi_\pidx > \xi_{\pidx'}$ for every $l \neq l'$ and $\bxi \in \Xi^L$, and we shall assume this in the remainder of the article.

Throughout this article, we shall be interested in making ordering statements among individual parameters, which we formally define as follows. 
\begin{definition}
An ordering statement is an event in $\B(\Xi^L)$ that establishes order relations among (a subset of) individual parameters. No order relations are assumed among the remaining parameters.
\end{definition}
A trivial ordering statement is the empty statement, that is, entire space $\Xi^L$. While this has posterior probability one, it provides no information among the relative orderings of any individual parameters. 
The simplest meaningful ordering statements are those which order only two individual parameters, for example, $\{\bxi \in \Xi^L : \xi_1<\xi_2, \, \xi_3,\dots,\xi_L \in \Xi\}$. For simplicity, we shall just write this as $(\xi_1<\xi_2)$ and not make explicit that $\xi_3,\dots,\xi_L \in \Xi$. 
An example of a strict ordering statement is $\xi_1 < \cdots < \xi_L$, which provides an ordering among all the individual parameters; this typically has negligible posterior probability. The above examples are of relatively simple ordering statements, and we shall construct more sophisticated ordering statements in the sequel.

In the remainder of this section, we fix the posterior distribution $\PP$ and construct statements based on it.

\subsection{From local to global statements} \label{sec.ranking_statements}

We begin by considering pairwise comparisons. 
We define events $E_{\pidx,\pidx'} = (\xi_\pidx > \xi_{\pidx'})$,  $1 \leq \pidx \neq \pidx' \leq \ntotal$ and call these elementary statements because they compare pairs of parameters. These will be used as building blocks in the sequel. In particular, $E_{\pidx,\pidx'}$ denotes the event that the $\pidx$th parameter is ordered higher than the $\pidx'$th parameter. We then choose a lower bound $(1-\alpha)$ for the posterior probability of the elementary statements for $\alpha \in [0,1]$, we define sets
\begin{align} \label{eq.Adef}
\begin{aligned}
\overline{\A}_{\pidx,\alpha}
& = 
\left \{ \pidx' \in \{1, \dots, L\} \, : \, \PP (E_{\pidx',\pidx}) > 1-\alpha \right \}, 
\\
\underline{\A}_{\pidx,\alpha}
& = 
\left \{ \pidx' \in \{1, \dots, L\} \, : \, \PP (E_{\pidx,\pidx'}) > 1-\alpha \right \}.
\end{aligned}
\end{align}
%
The set $\overline{\A}_{\pidx,\alpha}$ consists of all parameters that are ordered higher than the $\pidx$th parameter with posterior probability at least $(1-\alpha)$, and the set $\underline{\A}_{\pidx,\alpha}$ consists of all parameters that are ordered lower than the $\pidx$th parameter with posterior probability at least $(1-\alpha)$; recall that we have fixed the posterior distribution before constructing statements.  In other words, the sets  $\underline{\A}_{\pidx,\alpha}$ and $\overline{\A}_{\pidx,\alpha}$ are sets of players whose posterior probability of having an ability different from that of player $l$ is high.
We order the $\pidx$th parameter only with respect to parameters in $\underline{\A}_{\pidx,\alpha}$ and $\overline{\A}_{\pidx,\alpha}$, since ordering  with respect to other parameters has a high degree of uncertainty.

We now define the following local statements which involve multiple elementary comparisons.  We use the notation $|\A|$ to denote the cardinality of a finite set $\A$.  
\begin{enumerate}
\item \textbf{Statement $A_{\pidx,\alpha}$}:
parameter $\pidx$ is ordered higher than all parameters in $\underline{\A}_{\pidx,\alpha}$ and ordered lower than all parameters in $\overline{\A}_{\pidx,\alpha}$, that is,
\begin{align*}
A_{\pidx,\alpha} 
& = 
(\medcap_{\pidx' \in \underline{\A}_{\pidx,\alpha} } E_{\pidx,\pidx'} ) \medcap ( \medcap_{\pidx' \in \overline{\A}_{\pidx,\alpha} } E_{\pidx', \pidx} ).
\end{align*}
In other words, $A_{\pidx,\alpha}$ is the event where all the elementary statements associated with $\underline{\A}_{\pidx,\alpha} \cup \overline{\A}_{\pidx,\alpha}$ simultaneously hold.

\item  
\textbf{Statement $A_{\pidx,\alpha, t}$}: at least $(1-t) \times 100$ percent of the elementary statements in $A_{\pidx,\alpha}$ are $\alpha$-correct for $t \in [0,1]$, where $\alpha$-correct refers to an elementary statement $E_{l, l'}$ or $E_{l',l}$ such that $\PP(E_{l, l'}) > (1-\alpha)$ or $\PP(E_{l',l}) > (1 - \alpha)$, respectively; we will refer to $t$ as the local error. 
%
%
Formally,
\begin{align} \label{eq.Adef2}
A_{\pidx,\alpha, t}
= 
\medcup_{ \scriptsize
(\underline{\A},\overline{\A}) \in \mathscr{A}_{\pidx,\alpha,t}
}
(\medcap_{\pidx' \in \underline{\A} } E_{\pidx,\pidx'} ) \medcap ( \medcap_{\pidx' \in \overline{\A}} E_{\pidx', \pidx} ),
\end{align}
where
\begin{equation*}
\mathscr{A}_{\pidx,\alpha,t} = 
\left\{ (\underline{\A},\overline{\A})\, :\,
\underline{\A} \subseteq \underline{\A}_{\pidx,\alpha}, \,
\overline{\A} \subseteq \overline{\A}_{\pidx,\alpha}, \,
|\underline{\A}| + |\overline{\A}| 
\geq (1-t) \times | \underline{\A}_{\pidx,\alpha} \cup \overline{\A}_{\pidx,\alpha} | 
\right\}.
\end{equation*}
Notice that $\underline{\A} \cup \overline{\A}$ with $(\underline{\A},\overline{\A}) \in \mathscr{A}_{\pidx,\alpha,t}$ is a collection of indices with cardinality greater than or equal to $(1-t) \times | \underline{\A}_{\pidx,\alpha} \cup \overline{\A}_{\pidx,\alpha} |$ and for which either $\PP (E_{\pidx',\pidx}) > 1-\alpha$ or $\PP (E_{\pidx',\pidx}) > 1-\alpha$.
\end{enumerate}

$A_{l,\alpha}$ is the subset of $\Xi^L$ where $\xi_{l'} < \xi_l$ for all $l' \in \underline{\A}_{l,\alpha}$ and $\xi_{l} < \xi_{l'}$ for all $l' \in \overline{\A}_{l,\alpha}$; where we note that the sets $\underline{\A}_{l,\alpha}$ and $\overline{\A}_{l,\alpha}$ were chosen based on the posterior distribution $\PP$.
This is relaxed in the event $A_{l, \alpha, t}$, which does not require all the elementary statements in $A_{l,\alpha}$ to hold.
The posterior probability of at least one of the elementary statements in $A_{\pidx,\alpha}$ being $\alpha$-incorrect is a non-decreasing function of the number of such statements in $A_{\pidx,\alpha}$, and therefore $\PP(A_{\pidx,\alpha})$ is a non-increasing function of the number of elementary statements in $A_{\pidx,\alpha}$. Since our final aim is to define statements having high posterior probability, statement $A_{\pidx,\alpha, t}$ considers the same elementary statements as $A_{\pidx,\alpha}$ but also allowing for a local error determined by $t$. 

To create statements involving multiple local statements for a fixed $t$ and $\alpha$, we consider only those players whose corresponding $A_{\pidx,\alpha, t}$ have posterior probability lower-bounded by $(1-\gamma)$ for $\gamma \in [0,1]$,
%
\begin{align} \label{eq.Gdef}
\G_{\alpha,t, \gamma}
& = 
\left \{ l \in \{1, \dots, L\} \, : \, \PP \left ( A_{\pidx,\alpha, t} \right ) \geq 1-\gamma \right \}.
\end{align}
We now define the following global statements. 
\begin{enumerate}
\item \textbf{Statement $G_{\alpha,t, \gamma}$}: statements 
$A_{\pidx,\alpha, t}$ simultaneously hold for all $\pidx \in \G_{\alpha,t, \gamma}$,  that is, $G_{\alpha,t, \gamma} = \medcap_{\pidx \in \G_{\alpha,t, \gamma}} A_{\pidx,\alpha, t}$.  In other words, $G_{\alpha,t, \gamma}$ is the event where all the local statements  associated with  $ \G_{\alpha,t, \gamma} $ simultaneously hold.

\item 
\textbf{Statement $G_{\alpha,t, \gamma, q}$}: 
at least $(1-q) \times 100$ percent of the local statements comprising $A_{\pidx, \alpha, t}$ are $\gamma$-correct for $q \in [0,1]$, where $\gamma$-correct refers to a statement $A_{l,\alpha,t}$ such that $\PP ( A_{\pidx,\alpha,t}) \geq 1-\gamma$; we will refer to $q$ as the global error. 
Formally, 
\begin{align} \label{eq.global_statement}
G_{t,\alpha,\gamma,q}
= 
\medcup_{ \scriptsize
\G \in \mathscr{G}_{\alpha,t,\gamma,q}
}
\medcap_{\pidx \in \G} A_{\pidx,\alpha, t},
\end{align}
where
\begin{equation*}
\mathscr{G}_{\alpha,t,\gamma,q} = 
\left\{ \G\, :\,
\G \subseteq \G_{\alpha,t, \gamma}, \,
|\G| \geq  (1-q) \times |\G_{\alpha,t, \gamma}|
\right\}.
\end{equation*}
Notice that $\G \subseteq \G_{\alpha,t, \gamma}$ is a collection of indices $\pidx'$ with cardinality greater  than or equal to $  (1-q)  \times |\G_{\alpha,t, \gamma}|$ and for which  $\PP \left ( A_{\pidx,\alpha, t} \right ) \geq 1-\gamma$.

\end{enumerate}
Recall that we choose the set $\G_{\alpha,t,\gamma}$ based on the posterior probability of events $A_{l,\alpha,t}$. Having chosen $\G_{\alpha,t,\gamma}$, $G_{\alpha,t, \gamma}$ is then the subset of $\Xi^L$ where the events $A_{l,t,\alpha}$ happen for all $l \in \G_{\alpha,t,\gamma}$. This is relaxed in the global statement $G_{\alpha,t, \gamma, q}$, which requires only a proportion of the events $A_{l,t,\alpha}$ for $l \in \G_{\alpha,t,\gamma}$ to happen.

The global statement $G_{\alpha,t, \gamma, q}$ allows a global error $q$ that plays a role similar to that of $t$ at the local level. These global statements are of principal interest as they can be employed to produce rankings. \cref{app:algorithm} provides an algorithm to find local and global statements along with the corresponding posterior probabilities using samples drawn from $\PP$. 

One of our goals is to produce global statements that balance sparsity and the desire for statements having high posterior probability,  which can be achieved through specific choices of $(\alpha,t, \gamma, q)$. The following inequalities provide intuition about the relation between $(\alpha,t, \gamma, q)$ and the posterior probability of global statements and local statements.
\begin{enumerate}
\item 
$\PP(G_{\alpha,t, \gamma, q}) \geq \PP(G_{\alpha,t, \gamma, q'})$ for $q\geq q'$, that is, global statements with smaller global error have smaller posterior probability.
\item
$\PP(G_{\alpha,t, \gamma, q=0}) = \PP(G_{\alpha,t, \gamma}) \geq 1 - \gamma |\G_{\alpha,t,\gamma}|$, that is, when the global error is equal to zero, the posterior probability of the global statement increases as $\gamma$ or the number of local statements in $G_{\alpha,t,\gamma}$ decreases.
\item
$\PP(A_{\pidx,\alpha,t}) \geq \PP(G_{\alpha,t, \gamma})$ if $\pidx \in \G_{\alpha,t, \gamma}$, that is,  when the global error is zero, the posterior probability of the global statement is upper bounded by the posterior probability of any local statement in $G_{\alpha,t, \gamma}$.
\item
$\PP(A_{\pidx,\alpha,t=0}) = \PP(A_{\pidx,\alpha}) \geq 1 - \alpha |\underline{\A}_{\pidx,\alpha} \cup \overline{\A}_{\pidx,\alpha}|$, that is, when the local error is equal to zero, the posterior probability of the local statement increases as $\alpha$ or the number of elementary statements in $A_{\pidx,\alpha}$ decreases.
\end{enumerate}
Moreover, for fixed $q$ and $t$, the number of elementary statements in  $\A_{\pidx,\alpha,t}$ and local statements in $\G_{\alpha,t, \gamma, q}$ decreases as either $\alpha$ or $\gamma$ decreases. We can thus obtain less sparse statements by setting $\alpha$ and $\gamma$ at large values, which in turn increases the posterior probability of the global statement. The roles of $q$ and $t$ are less obvious.  For example, if we increase $q$ and $t$, it is not clear how this will impact the posterior probability of $G_{\alpha,t, \gamma, q}$.  On one hand,  the posterior probability of the local statements will increase leading to less sparse global statements which may decrease its posterior probability; while on the other hand, an increase of $q$ will also allow more errors which may result in an increase in the posterior probability. To deal with the effect of variations of $(\alpha,t, \gamma, q)$ on the posterior probability and on sparsity of $\G_{\alpha,t, \gamma, q}$, we introduce a decision theoretic framework in the next \cref{sec.decision} to select values of $(\alpha,t, \gamma, q)$, and illustrate their impact on the ordering statements in \cref{sec.application}.  

\subsection{Decision theoretic framework} \label{sec.decision}

In the previous section, we defined a hierarchy of ranking statements (elementary, local, and global). Our goal is to identify global statements having high posterior probability while simultaneously including as many local statements as possible.  To this end, we formulate this as an optimization problem within a decision theoretic framework. 
We use a reward function and assume that the space of actions corresponds to the set of all possible global statements, that is, each action is equivalent to a global statement. We can in principle construct global statements by putting together any arbitrary combination of elementary and local statements without using a posterior distribution. Since the space of actions is finite, we should then compute the reward function for each of these statements and choose the one that maximises the reward. However, the space of actions is massive, and it turns out to be infeasible to evaluate the reward function for all actions. For computational efficiency, we discard in advance actions that have low posterior probability. This is why we consider global statements constructed using the posterior distribution, that is, statements of the form $G_{\alpha,t, \gamma, q}$ with non-negligible posterior probability. We remark that since we only consider global statements constructed using the posterior distribution (in particular, using the data), our strategy is not strict decision theoretic approach, but rather an approximate one.


When optimizing the reward function, locally, we seek statements that maximise the number of elementary comparisons, which is given by $|\underline{\A}_{\pidx,\alpha} \cup \overline{\A}_{\pidx,\alpha}|$ for the $l$th parameter.
However, each local statement $A_{\pidx,\alpha,t}$ allows an error that depends on $t$, and we seek to simultaneously minimize this.  Globally, we seek to maximize the number of local statements $|\G_{\alpha, t, \gamma}|$ while minimizing the error $q$. Finally, we want the posterior probability of the final global statement $G_{\alpha, t, \gamma, q}$ to be high. With this in mind and defining $a = (\alpha, t, \gamma, q) \in [0,1]^4$ to be an action (that in turn leads to a statement $G_{\alpha,t, \gamma, q}$), we couch this in a decision theoretic framework and seek to design a reward function satisfying the following \cref{reward_cond}.

\begin{condition} \label{reward_cond}
The reward function assessing $G_{\alpha, t, \gamma, q}$ satisfies the following.

\begin{enumerate}
\item 
It is a non-decreasing function of $|\underline{\A}_{\pidx,\alpha} \cup \overline{\A}_{\pidx,\alpha}|$  for $\pidx \in \G_{\alpha,t,\gamma}$, and of $|\G_{\alpha,t,\gamma}|$.

\item 
It is a non-increasing function of $t$ and $q$.

\item 
It is a non-decreasing function of $\PP(G_{\alpha, t, \gamma, q})$.

\end{enumerate}
\end{condition} 

An example of such a reward function is 
\begin{align} \label{eq.reward_fn}
R(a) 
= 
R (\alpha, t, \gamma, q) 
= C(G_{\alpha,t,\gamma,q})\times
\PP(G_{\alpha,t,\gamma,q}),
\end{align}
where 
\begin{align} \label{eq.cost}
\begin{aligned}
C(G_{\alpha,t,\gamma,q}) 
& =
\lfloor (1-q) \times |\G_{\alpha,t, \gamma}|\rfloor 
\times 
\left \{ \sum_{\pidx \in \G_{\alpha,t,\gamma}}
\lfloor (1-t) \times |
\underline{\A}_{\pidx,\alpha} \cup \overline{\A}_{\pidx,\alpha}
| \rfloor
\right \}
\end{aligned}
\end{align}
satisfies the first two parts of \cref{reward_cond}. In this case, the factor $C(G_{\alpha,t,\gamma,q})$ equals the $(1-q)\times 100$ percent of the number of local statements within $G_{\alpha,t,\gamma,q}$ multiplied by $(1-t)\times 100$ percent of the total number of elementary statements within the local statements comprising $G_{\alpha,t,\gamma,q}$. Thus, adding (or removing) local statements in $G_{\alpha,t,\gamma,q}$ will raise (or lower) $C(G_{\alpha,t,\gamma,q})$ by $(1-q)\times 100$ percent of the number of added (or removed) local statements multiplied by $(1-t)\times 100$ percent of the total number of elementary statements within the added local statements. If the number elementary statements increases (or decreases) within a local statement already in $G_{\alpha,t,\gamma,q}$, the factor $C(G_{\alpha,t,\gamma,q})$ will change by $(1-t)\times 100$  percent of the number of added (or removed) elementary statements.   It is worth noting that the reward function in \cref{eq.reward_fn} changes in direct proportion to $\PP(G_{\alpha,t,\gamma,q})$ and $C(G_{\alpha,t,\gamma,q}) $. If either $C(G_{\alpha,t,\gamma,q})$ or $\PP(G_{\alpha,t,\gamma,q})$ change in $\pm 1$\%, the reward function will also change in $\pm 1$\%. Changes in reward function values are primarily governed by the trade-off between $C(G_{\alpha,t,\gamma,q}) $ and $\PP(G_{\alpha,t,\gamma,q})$ because less sparse global statements (i.e., larger $C(G_{\alpha,t,\gamma,q}) $) leads to smaller posterior probabilities (i.e., smaller $\PP(G_{\alpha,t,\gamma,q})$). Depending on the application, the dynamics of this trade-off will differ.

With an abuse of notation, the function $C$ is interpreted as a measure of the quality of a global statement $G$ comprised of local statements as described in \cref{sec.ranking_statements}; this does not necessarily have to be constructed using the parameters $(\alpha, t, \gamma, q)$. Restricting ourselves to the set of global statements that are indeed constructed using parameters $(\alpha, t, \gamma, q)$, the reward function \eqref{eq.reward_fn} corresponds to the posterior mean of a utility function of the form $u\{(\xi_1, \dots, \xi_\ntotal), a )\} =  C(G_{\alpha, t, \gamma,q})\times
\indicator ( G_{\alpha, t, \gamma,q} )$. Doing this allows us to reduce the space of all possible global statements to a more manageable size and constrain the optimization to statements having high posterior probability. 

An efficient algorithm to approximate the optimizer of the reward function $\widehat{a} = (\widehat{\alpha}, \widehat{t}, \widehat{\gamma}, \widehat{q}) = \argmax_{a \in [0,1]^4} R(a)$ is provided in \cref{app:opt}. \cref{sec.asynptotic} provides asymptotic guarantees showing that, under this strategy, the optimal global statement $G_{\widehat{\alpha}, \widehat{t}, \widehat{\gamma}, \widehat{q}}$ converges to the ``true'' ranking. Since users may also be interested in fixing local and global errors at desired values or making statements that ensure a minimum posterior probability, the algorithm in \cref{app:opt} also enables users to approximate optimal global statements over any subset of $(\alpha, t, \gamma, q)$ such that $\PP(G_{\alpha, t, \gamma, q})$ is above a given value.

The computational cost of finding a global statement given $M$ samples drawn from the posterior distribution is $\OO(M \ntotal^2)$; this is because we have to consider all possible pairwise comparisons to form the sets $E_{\pidx,\pidx'}$, $\overline{\A}_{\pidx,\alpha}$ and $\underline{\A}_{\pidx,\alpha}$. However, this can be trivially parallelised since they can be computed independently of each other for each parameter. Having written code in \textsf{R} and \CC, even without parallelisation, it takes only a few minutes to compute global statements for up to $4 \times 10^3$ parameters and $2 \times 10^3$ posterior samples. Precise results for computational times required to approximate the optimal global statements are provided in \cref{sec:simulation.resuls} for a simulation study, and in Sections \ref{sec:results_lineups} and \ref{sec:results_players} for the NBA application.

\subsection{Large data limit} \label{sec.asynptotic}

We study the asymptotic behavior of the proposed procedure when there is a ``true'' underlying $\bxi^\star = (\xi_1^\star, \dots, \xi_{\ntotal}^\star)$ that generates the data and the probability measure $\PP$ concentrates around this truth as the amount of data increases. This is the case when $\PP$ is a consistent posterior distribution. 
Writing the measure $\PP$ as a function of $N$ as $\PP^{(N)}$, where $N$ is the number of observations, we assume that the following consistency condition holds:
\begin{align} \label{eq.consistency}
\PP^{(N)} \left \{ \bxi \in B_\epsilon(\bxi_\star) \right \}
\rightarrow 
1 
~~ \text{in}~ \mathbb{P}_{\bxi^\star}\text{-probability for any} ~ \epsilon > 0,
\end{align}
where $B_\epsilon(\bxi_\star)$ denotes a ball around $\bxi^\star$ of radius $\epsilon>0$, $B_\epsilon(\bxi_\star) = \left \{ \bxi \in \totspace^\ntotal \, : \, \dist (\bxi, \bxi_\star) \leq \epsilon \right \}$, $\dist(\cdot, \cdot)$ denotes a distance measure on $\totspace^\ntotal$, and $\mathbb{P}_{\bxi^\star}$ denotes the data-generating mechanism under parameter value $\xi^\star$. In the following \cref{thm.asymptotic}, we show that our method is guaranteed to produce a traditional ranking with high posterior probability as $N \rightarrow \infty$ while assuming that, without loss of generality, $\xi^\star_1 \prec \cdots \prec \xi^\star_{\ntotal}$.

\begin{proposition}\label{thm.asymptotic}
Let $\bxi^\star$ be such that $\xi^\star_1 \prec \cdots \prec \xi^\star_{\ntotal}$. If the distribution $\PP^{(N)}$ concentrates at $\bxi^\star$ in $\mathbb{P}_{\bxi^\star}$-probability as $N \to \infty$ (that is, \cref{eq.consistency} holds), the optimal action $(\widehat{\alpha}^{(N)}, \widehat{t}^{(N)}, \widehat{\gamma}^{(N)}, \widehat{q}^{(N)})$ converges to $(\widehat{\alpha}^{(\infty)}, \widehat{t}^{(\infty)}, \widehat{\gamma}^{(\infty)}, \widehat{q}^{(\infty)}) \equiv (0,0,0,0)$ in $\mathbb{P}_{\bxi^\star}$-probability, with $G_{0,0,0,0}$ stating that $\xi_1 \prec \cdots \prec \xi_{\ntotal}$.
\end{proposition}  

\begin{proof}[Proof of \cref{thm.asymptotic}]
We can find $\epsilon>0$ such that $\xi_1 \prec \cdots \prec \xi_{\ntotal}$ for all $\bxi \in B_\epsilon(\bxi^\star)$, for example, by choosing $\epsilon = \min_{\pidx \in \{1, \dots, \ntotal\}} |\xi_{\pidx+1} - \xi_\pidx|/3$ if $\xi_\pidx \in \RR$ for all $\pidx$. By \cref{eq.consistency}, for any $\delta>0$, we can choose $N$ large enough such that 
\begin{align} \label{eq.proof_1}
\PP^{(N)} \left \{ B_\epsilon(\bxi^\star) \right \} 
& > 
1-\delta,
\end{align}
which implies $\PP^{(N)}(\xi_1 \prec \cdots \prec \xi_{\ntotal}) > 1-\delta$. 
Choosing 
\begin{align} \label{eq.proof_2}
\delta 
& < 
\min \{\alpha, \gamma\} 
\end{align}
implies that $\PP^{(N)}(E_{\pidx,\pidx'}) > 1-\delta$ for $\pidx > \pidx'$ and $\PP^{(N)}(E_{\pidx',\pidx}) > 1-\delta$ for $\pidx < \pidx'$. 
Since $\delta < \alpha$, from \cref{eq.Adef2} we have that $A_{\pidx,\alpha,t} \supseteq B_\epsilon(\bxi^\star)$ for all  $\pidx \geq 0$, and thus $\PP^{(N)}(A_{\pidx,\alpha,t}) \geq \PP^{(N)}\{B_\epsilon(\bxi^\star)\} \geq 1-\delta$ for all  $\pidx \geq 0$.  
Since $\delta < \gamma$, this implies that $\G_{\alpha,t, \gamma} = \{1, \dots, \ntotal\}$ for any $\gamma >0$.
This in turn implies that $G_{\alpha, t, \gamma, q} \supseteq B_\epsilon(\bxi^\star)$ for all $q \geq 0$. 
Thus the reward function \eqref{eq.reward_fn} is maximized for $t=q=0$. 
Since the reward function \eqref{eq.reward_fn} is a decreasing function of $\alpha$ and $\gamma$, and since $\PP(G_{\alpha, t, \gamma, q}) \geq 1-\delta$, the reward function is maximized for $a = (\widehat{\alpha}^{(\infty)}, \widehat{t}^{(\infty)}, \widehat{\gamma}^{(\infty)}, \widehat{q}^{(\infty)})$. The result now follows from equations \eqref{eq.proof_2} and \eqref{eq.proof_1}.
\end{proof}

\section{Simulation study}\label{sec:simulation_study}

We perform a simulation study to investigate the global statements defined in \cref{sec:contributions}.  The simulation study is based on a subset of the NBA lineup data so that synthetic datasets mimic the inherent complexities of the NBA lineup data.
\cref{sec:data.model} introduces the approach adopted to model the data (which is used later in \cref{sec.application}) and generate synthetic datasets. \cref{sec:simulationsetting} describes simulation of the synthetic datasets and factors considered in the simulation study, and \cref{sec:simulation.resuls} presents the obtained results.  

\subsection{Statistical model} \label{sec:data.model}

The model used to estimate player abilities based on the NBA data is detailed; this will be used in the sequel to illustrate our approach. We emphasize that the framework of  \cref{sec:contributions} is applicable regardless of the model posed and the dataset considered.

Let $z_i$ and $z_i'$ denote the points scored by the two lineups that competed in the $i$th encounter, where $i = 1, \ldots, N$, and $N$ denotes the number of encounters. Each lineup is comprised of five players and, as before, we assume that $\bxi = (\xi_1, \dots, \xi_{\ntotal}) \in \RR^{\ntotal}$ denotes the abilities of the $\ntotal$ total players.
Let $\L_i \subset \{1, \dots, \ntotal\}$ and $\L'_i \subset \{1, \dots, \ntotal\}$ such that $\L_i \cap \L'_i = \emptyset$ identify the five players of the two lineups in the $i$th encounter.  The ability of lineup $\L_i$ is defined as
\begin{align} \label{eq.lineup_ability}
\tau_{\L_i}
& = 
\sum_{k \in \L_i} \xi_{k}.
\end{align}
By defining the ability of lineup $\L_i$ as the sum of the abilities of individuals in the lineup, we induce correlation among lineups that share players and achieve a more parsimonious setup. Following \cite{huang2008ranking}, we model the difference in points scored as 
\begin{align} \label{eq.gaussian_model}
(z_i - z_i') \mid \bxi, \sigma^2 
& \ind
{\rm Normal}(
\tau_{\L_i} - \tau_{\L'_i}, \sigma^2
), ~~ i = 1, \ldots, N.
\end{align}
We estimate $\bxi$ and $\sigma^2$ under a Bayesian approach, and choices for the prior distribution of $\bxi$ and $\sigma^2$ are discussed in \cref{sec:simulationsetting}.  Notice that \eqref{eq.lineup_ability} and \eqref{eq.gaussian_model} correspond to an extension of the Bradley-Terry model \citep{bradley1952rank} for continuous scores and that includes group information (i.e., lineups). If our goal were to create ordering statements for lineups only, a potential modeling strategy would be to have lineup-specific parameters. However, this would result in a less parsimonious model, which would introducing more uncertainty, and prediction for unobserved lineups would not be available.

We note that the model given by equations \eqref{eq.lineup_ability} and \eqref{eq.gaussian_model} is not identifiable. Different conditions can be imposed for identifiability, such as requiring $\ntotal^{-1} \sum_{\pidx=1}^{\ntotal} \xi_\pidx = 0$, which is interpreted as an average player having ability zero, or requiring $\xi_{\widetilde{\pidx}}=0$ for a fixed $\widetilde{\pidx} \in \{1, \dots, \ntotal\}$, in which case all other abilities are interpreted as relative to the ability of the $\widetilde{\pidx}$th player. For this paper, we impose the latter constraint with $\widetilde{\pidx}$ being the player with the highest frequency of appearance. 

We acknowledge that the structure of the NBA data is particularly complex and model \eqref{eq.gaussian_model} does not account for some of the structure present in the data. For example, point differences are discrete random variables, player abilities are expected to be correlated with playing times, some players might have similar abilities (i.e., the modeling strategy could account for ties), information beyond lineup composition (e.g., homecourt) could influence point differences, and the definition of lineup abilities should include interaction terms between players.  Nonetheless, this simple model is expected to be able to capture important orderings among players' and lineups' abilities. 

\subsection{Simulation setting} \label{sec:simulationsetting}

We conduct a simulation study to examine the performance of the proposed method under four different conditions: (i) the prior distributions for $\bxi$ and $\sigma^2$, (ii) the sample size (which is given by the number of encounters), (iii) the true abilities to be estimated, and (iv) the number of parameters to be ordered (which is given by the number of players). In what follows, we provide specific details about how synthetic datasets are generated; as a high level description, our approach is to first fit model \eqref{eq.gaussian_model} to a subset of the NBA data and then employ the posterior distribution of $\bxi$ and $\sigma^2$ as the seed to generate true player abilities.

To complete the model specification described in the previous section, we consider the following priors for $\bxi$ and $\sigma^2$.

\begin{tabular}{ll}
Prior 1:  &  $\xi_\pidx \mid \mu \iid \lap(\mu,3), \quad \pidx = 1, \dots, \ntotal,$ \\
Prior 2:  &  $\xi_\pidx \mid \mu \iid \lap ( \mu, \sigma \lambda^{-1/2}), \quad \pidx = 1, \dots, \ntotal, \quad
\lambda  \sim \mathrm{Gamma}(1,1),$\\
Prior 3:  & $\xi_\pidx \mid \mu \iid \Nor(\mu,\sigma^2\lambda), \quad \pidx = 1, \dots, \ntotal, \quad \lambda \sim \mathrm{Half}\mbox{-}\mathrm{Cauchy}(0,1)$,
\end{tabular}

\noindent with $\pi (\sigma^2)  \propto \sigma^{-2}$ and  $\mu \sim \Nor(0,3^2)$. Prior 1 represents an informative scenario that assumes knowledge of the approximate variability of $\xi_\pidx$s. Priors 2 and 3 are less informative and induce shrinkage towards the mean $\mu$, with prior 3 producing more shrinkage than prior 2.  Prior 2 is motivated by \cite{park2008bayesian} and \cite{hans2009bayesian,hans2010model}, and prior 3  by \cite{gelman2006prior} and \cite{polson2012bayesian}.  The prior variance for $\mu$ was set to $9$ due to the fact that most encounters between lineups are competitive and so few result in point differences of double digits.

Since the idea is to simulate data that mimic the complexities of the NBA lineup data, we consider two subsets of data with five and ten randomly selected teams, respectively. The subset with five teams has 78 players and 844 encounters, while the one with ten teams has 167 players and 3604 encounters. For both subsets, we sample from the posterior distribution of the abilities of players based on the three priors described above. We use Stan \citep{stan} to collect $10^4$ samples from each posterior distribution, from which we discard the initial 20\% as burn-in and thin every 8 after that.  We use the no-U-turn sampler (NUTS) \citep{hoffman2014no}, which is an adaptive version of Hamiltonian Monte Carlo \citep{duane1987hybrid}. From the $10^4$ samples, we randomly select fifty samples denoted by $(\tilde\bxi_{(j,k,d)}, \tilde\sigma^2_{(j,k,d)})$, with $j=1, \dots, 50$, indexing the samples, $k=1,2,3$ indexing the prior used for data generation, and $d=5,10$ being the number of teams. The use of priors 1, 2, and 3 for data generation offers different levels of similarity among true abilities, with prior 1 and 3 leading to the least and most similar abilities, respectively.

We generate synthetic datasets as follows. For each value of $(j,k,d)$, we assume that $\tilde\bxi_{(j,k,d)}$ and $\tilde\sigma^2_{(j,k,d)}$  are the true abilities and variance, respectively. Given the shrinkage properties of the considered priors, the abilities in $\tilde\bxi_{(j,3,d)}$ are more similar to each other than those in $\tilde\bxi_{(j,2,d)}$, and $\tilde\bxi_{(j,1,d)}$ represents the case where abilities are the least similar. Using model \eqref{eq.gaussian_model}, we simulate differences in points scored assuming that the lineups in each encounter are the same as observed in the NBA data. The synthetic datasets thus are comprised of 844 (for $d=5$) and 3604 (for $d=10$) encounters. Moreover, for each $(j,k,d)$, we consider two sets of experiments indexed by $s = 1, 2$. In the first set ($s=1$), we consider the encounters as they are, while in the second set ($s=2$), we generate two independent realizations of the differences in points scored $z_i$ for each encounter. Thus for the second set of experiments, the synthetic datasets are comprised of 1688 (for $d=5$) and 7208 (for $d=10$) encounters. In total, we are generating 600 ($50 \times 3 \times 2 \times 2$) synthetic datasets, where each of them are associated with a particular combination of $(j,k,d,s)$. We use $\mathrm{D}^{(j,k,d,s)}$ to denote the synthetic dataset corresponding to combination $(j,k,d,s)$.

For each synthetic dataset we sample from the posterior distributions of the abilities of players and lineups using the three priors described above. We use the algorithm in \cref{app:opt} and restrict the search to statements with posterior probability greater than $0.9$ to obtain the optimal global statements.  
Let $G_{\widehat{a}}^{(j,k,d,s,m)}$ denote the optimal global statement -- note that $\widehat{a} = (\widehat{\alpha}, \widehat{t},\widehat{\gamma},\widehat{q})$ is also a function of ${(j,k,d,s,m)}$ -- obtained from model \eqref{eq.gaussian_model}, prior $m$, and dataset $\mathrm{D}^{(j,k,d,s)}$. To evaluate the performance of our procedure at each combination $\bv = (k,d,s,m)$, we compute the following measures: (i) posterior probability of the global statements; (ii) frequentist coverage, that is, fraction of global statements that are consistent with the true abilities; (iii) the number of local statements and average number of players (or lineups) within local statements for each global statement; and (iv) local and global errors. 

We compare our proposed approach with the cumulative probability consensus ranking (CPCR) methodology proposed by \cite{vitelli2018probabilistic}. Because the information provided by our statements and CPCR cannot be compared directly (our approach produces a joint statement while the CPCR produces a type of marginal statement), we define a type of joint CPCR.  This permits us to compare the two approaches which we do based on measures equivalent to (i) and (ii), i.e., posterior probability of a type of joint CPCR; (ii) fraction of joint CPCRs that are consistent with the true abilities. \cite{vitelli2018probabilistic} define CPCRs as follows: first select the item which has the maximum a posteriori marginal probability of being ranked 1st--we denote this item with $\xi_{(1)}$; then the item which has the maximum a posteriori marginal posterior probability of being ranked 1st or 2nd among the remaining ones--denoted by $\xi_{(1:2)}$, etc. Thus, we define joint CPCRs for players as $\bigcap_{l=L_0}^L [\xi_{(1:l)}\mbox{ ranked $1$st, $2$nd, $\ldots$, or $l$th}]$, $L_0 = 1, \ldots, L$, and for lineups as $\bigcap_{l=L_0}^T [\tau_{(1:l)}\mbox{ ranked $1$st, $2$nd, $\ldots$, or $l$th}]$, $L_0 = 1, \ldots, T$, where $T$ is the number of observed lineups.

\subsection{Results}

\label{sec:simulation.resuls}

We apply our approach using the reward function in \cref{eq.reward_fn} to find global statements for lineups and players. For lineups, we draw a sample from the posterior distribution of player abilities, which simultaneously leads to a sample from the posterior distribution of lineup abilities $\tau_{\L_i}$.  The optimization described in \cref{sec.decision} is not performed over all 5-tuples, but only over the observed lineups. Given the similarities between the patterns observed for players and lineups, we report the results for lineups in \cref{fig:simulation.results}, while results for players are deferred to Figure \ref{fig:simulation.resultsS1} of the online Supplementary Material. Local and global errors for statements associated with players and lineups are both reported in the online Supplementary Material (see Figure \ref{fig:simulation.resultsS2}). Moreover, statements associated with lineups are based on the lineups from the Cleveland Cavaliers and San Antonio Spurs. 

\begin{figure}
\centering
\includegraphics[scale=0.45, trim={0 0cm 4.5cm 0cm},clip=true,page=1]{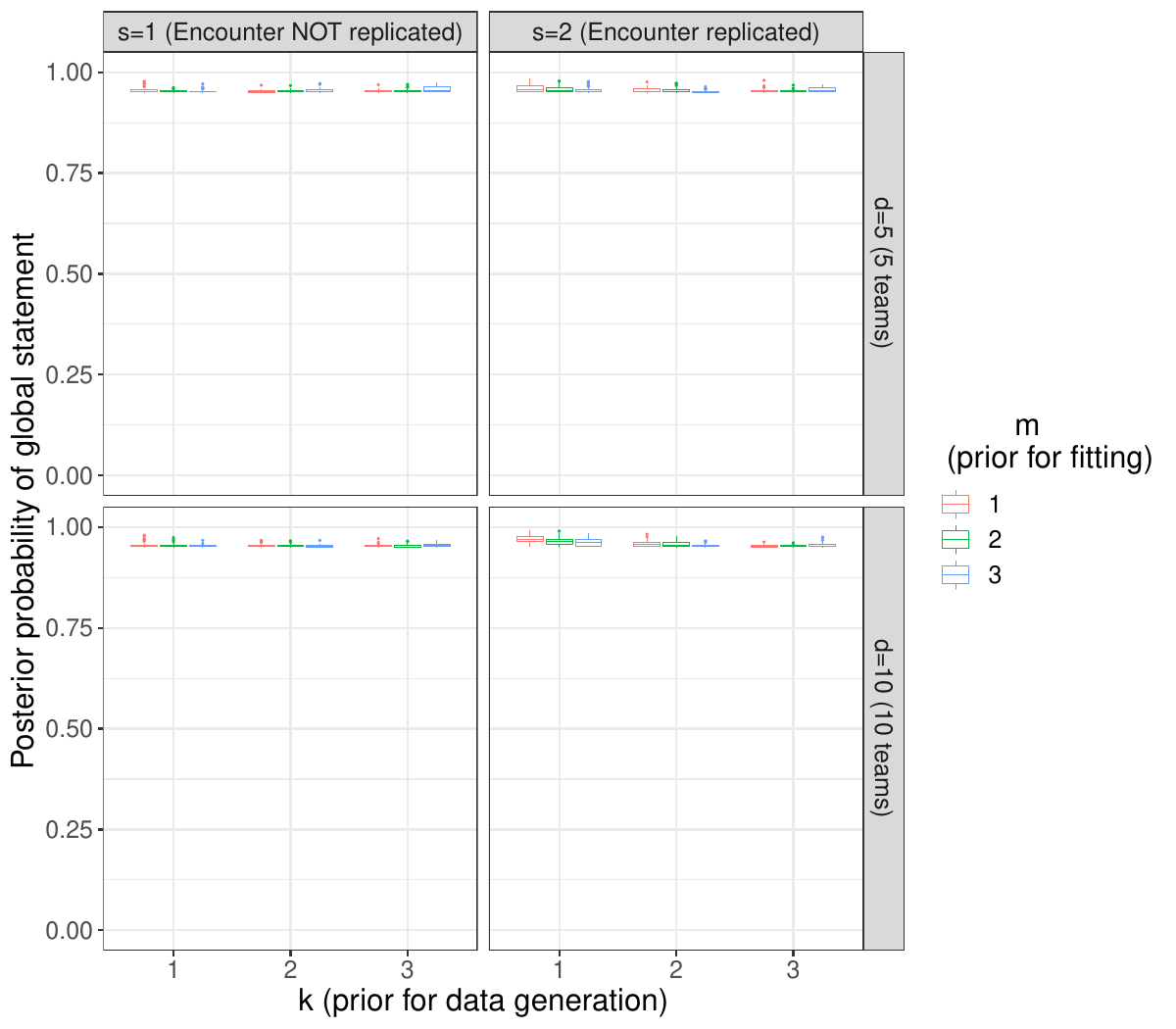}
\includegraphics[scale=0.45, trim={0 0cm 0cm 0cm},clip=true,page=2]{plots/FigureLineups_JRSSC.pdf}
\includegraphics[scale=0.45, trim={0 0cm 4.5cm -1cm},clip=true,page=3]{plots/FigureLineups_JRSSC.pdf}
\includegraphics[scale=0.45, trim={0 0cm 0cm -1cm},clip=true,page=4]{plots/FigureLineups_JRSSC.pdf}
\caption{Simulation results for lineups assessing the effect of prior distribution used for fitting (captured by $m$), prior distribution used for data generation (captured by $k$), number of parameters (captured by $d$), and doubling the sample size (captured by $s$) with $\bv = (k,d,s,m)$.  
The top-left plot displays the posterior probability of the global statements;
the top-right plot shows the frequentist coverage;
the bottom-left plot presents the number of local statements within global statements;
and the bottom-right plot shows the distribution of the average number of lineups within local statements across global statements. These plots summarize results over 50 synthetic datasets.}
\label{fig:simulation.results}
\end{figure}

The first observation from \cref{fig:simulation.results} is that the proposed procedure is, as expected, able to find global statements having high posterior probability (above $0.9$). The number of teams $d \in \{5,10\}$ does not appear to have a direct impact on the posterior probabilities or frequentist coverage as no clear patterns emerge. However, as expected, relative differences between true player abilities (reflected by $k$), coupled with the prior distribution specification $m$, impacts the frequentist coverage as well as the number of local statements and average number of lineups within local statements.  

Generally speaking, the average number of lineups within local statements are smallest for prior 3 (that is, $m=3$). Since making statements with a large number of lineups is desirable, this is unsatisfactory if the frequentist coverage does not remain high.  Notice, however, that as the relative similarity among player and lineup abilities increases (that is, $k$ increases), the coverage associated with priors 1 and 2 decreases quite dramatically even with sparse global statements (that is, low number of local statements and low average number of players or lineups within local statements); this too is undesirable. Thus, the balancing act alluded to earlier (correctness of statement versus the number of items the statement contains) is impacted by the prior distribution and  similarity of true parameter values. Indeed, when the true abilities are most similar ($k=3$) and the model accurately captures this ($m=3$), we observe global statements that are sparser, indicating that the method refrains from collecting many local statements when forming the global statement, which is the right decision in such cases. Finally, increasing the number of encounters (moving from $s=1$ to $s=2$) has the desired effect of producing statements that are less sparse and have higher posterior probability and smaller local and global errors (see Figure \ref{fig:simulation.resultsS2} of the Supplementary Material).

We now analyse the results obtained under the CPCR methodology. As originally intended by \cite{vitelli2018probabilistic}, the CPCR approach identifies, for each $l = 1,\ldots,T$, events of the form $[\tau_{(1:l)}\mbox{ ranked $1$st, $2$nd, $\ldots$, or $l$th}]$ that concentrates the highest posterior probability (see Figure \ref{fig:simulation.results.CPCR} in the supplementary material). Both the average posterior probability and the frequentist coverage for each CPCR are similar to each other across all combinations of simulation factors. The CPCRs representing the lineups with highest ability have probabilities less than or around $0.3$, while CPCRs representing lineups that have the lowest ability have probabilities close to one (by design CPCRs for the worst lineups are more imprecise). As mentioned, global statements and CPCR are not directly comparable because the former provides joint statements and the latter provides a type of marginal statement. For this reason, we consider joint CPCRs defined as $\bigcap_{l=L_0}^T [\tau_{(1:l)}\mbox{ ranked $1$st, $2$nd, $\ldots$, or $l$th}]$, $L_0 = 1, \ldots, T$. Figure \ref{fig:simulation.results.lineups.joint.CPCR} shows the effect that $L_0$ has on the probabilities for joint CPCRs. The larger $L_0$ (less lineups in the joint statement), the higher the probability. Notice how quickly these probabilities decay to zero. In fact, if the goal is to report a joint statement that concentrates a probability close to $0.9$, we must select $L_0 = 243$ for $d = 5$ ($5$ teams = $247$ lineups) and  $L_0 = 422$ for $d = 10$ ($10$ teams = $425$ lineups). For the case $L_0 = 243$ with $d = 5$, the most precise statement in $\bigcap_{l=243}^{247} [\tau_{(1:l)}\mbox{ ranked $1$, $2$, $\ldots$, or $l$}]$ states that lineup $\tau_{(1:243)}$ is ranked $1$st, $2$nd, $\ldots$, or $243$th out of $247$ lineups, which does not provide much information. Notice that under the CPCR methodology the most informative statements are those for the best lineups suggesting we could alternatively consider joint statements of the form $\bigcap_{l=1}^{L_0} [\tau_{(1:l)}\mbox{ ranked $1$, $2$, $\ldots$, or $l$}]$, $L_0 = 1, \ldots, T$. Unfortunately, the probability of any of those joints statements will be upper bounded by the probability of $[\tau_{(1)}\mbox{ ranked first}]$ which is, in the best case, around $0.3$ as stated before. This strategy would limit the possibility of finding joint statements that concentrate high probability. The results for players under the CPCR methodology are similar to those obtained for lineups (see Figures \ref{fig:simulation.results.CPCR} and \ref{fig:simulation.results.players.joint.CPCR} in the supplementary material).

\begin{figure}
\centering
\includegraphics[scale=0.45, trim={0 0cm 4.5cm 0cm},clip=true,page=3]{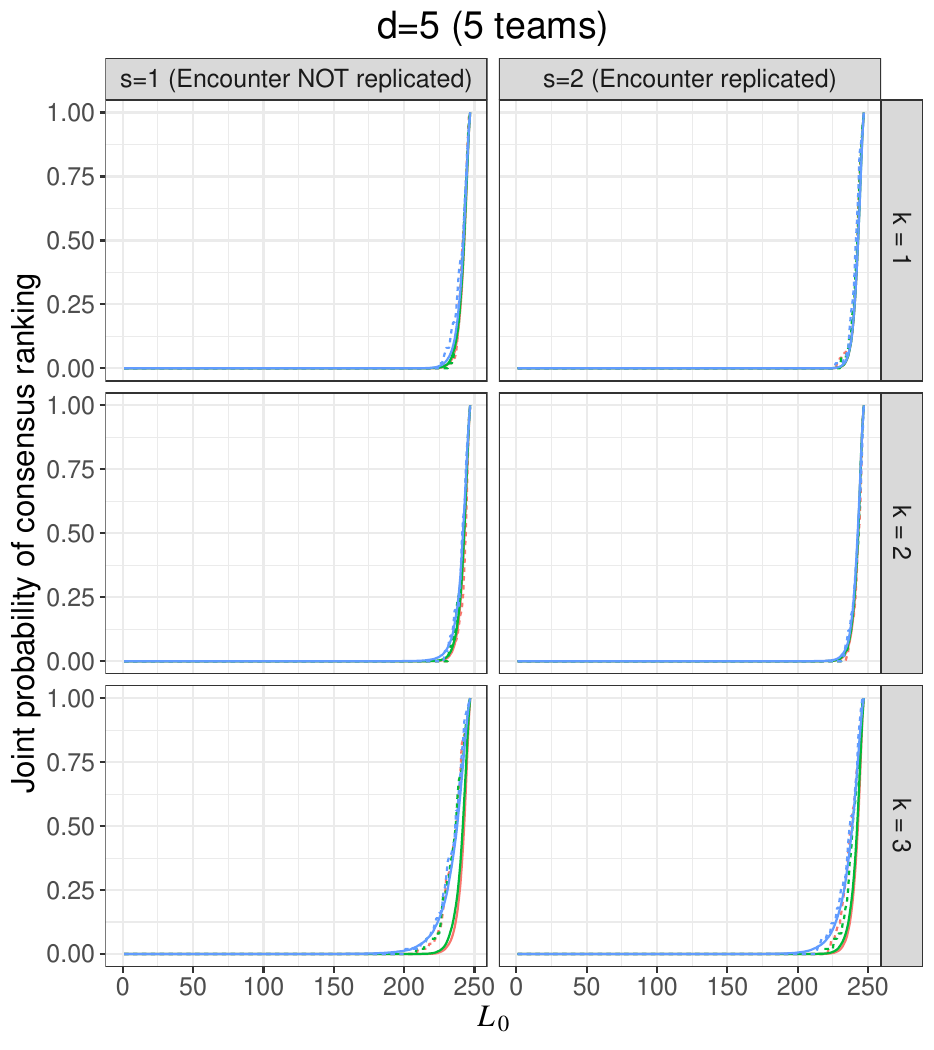}
\includegraphics[scale=0.45, trim={0 0cm 0cm 0cm},clip=true,page=4]{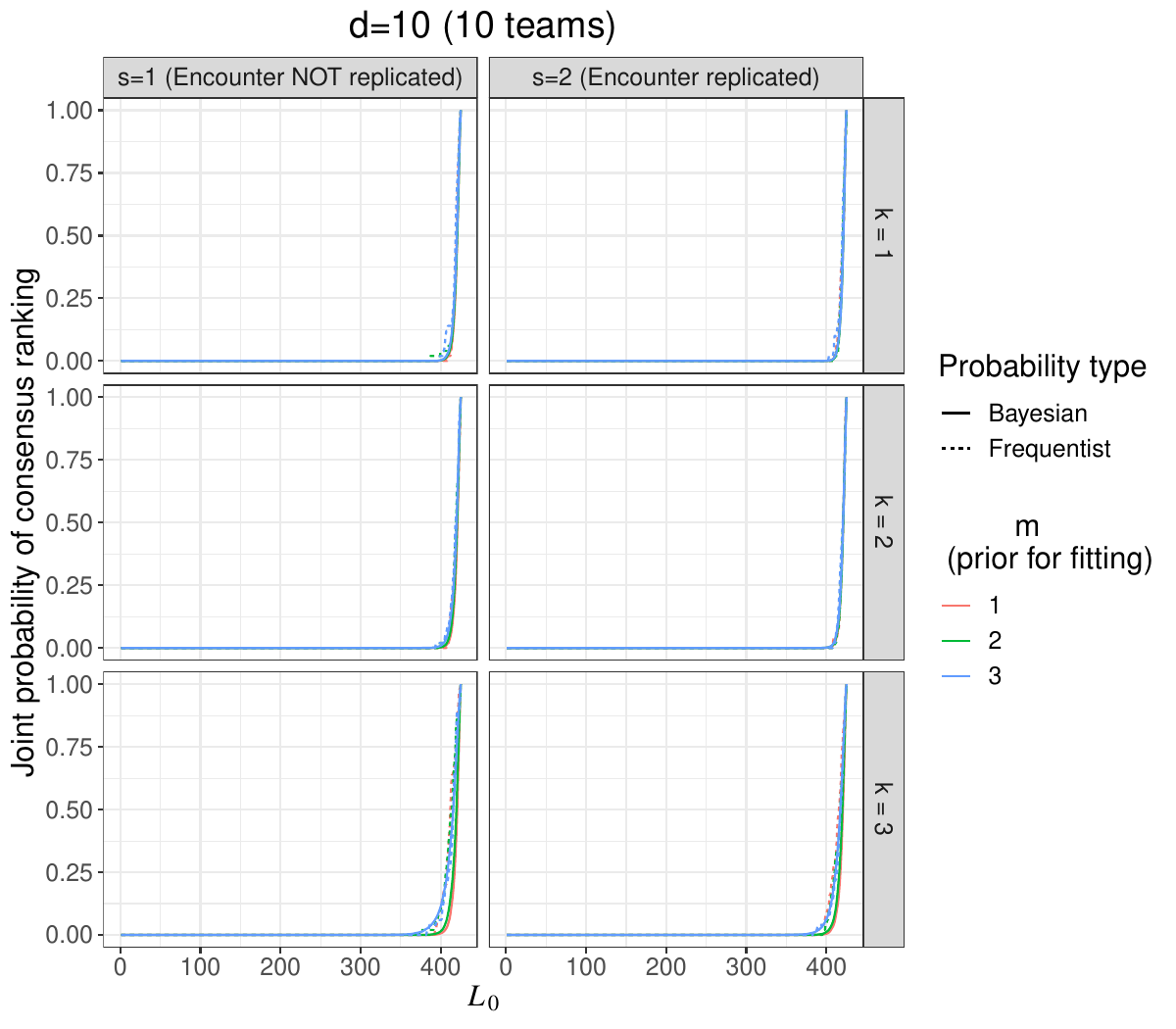}
\caption{Simulation results for joint CPCRs for lineups assessing the effect of prior distribution used for fitting (captured by $m$), prior distribution used for data generation (captured by $k$), number of parameters (captured by $d$), and doubling the sample size (captured by $s$) with $\bv = (k,d,s,m)$.  Each plot displays the average posterior probability and frequentist coverage of the joint CPCRs (i.e., $\bigcap_{l=L_0}^T [\tau_{(1:l)}\mbox{ ranked $1$, $2$, $\ldots$, or $l$}]$). These plots summarize results over 50 synthetic datasets.}
\label{fig:simulation.results.lineups.joint.CPCR}
\end{figure}

The main conclusion from the simulation study is that the proposed method  succeeds in finding statements having high posterior probability. The accuracy of the statements at containing the truth relies on the accuracy of the posterior distributions, and we observe that employing prior 3 produces global and local statements that maintain good frequentist properties at the cost of producing conservative statements.  On the other hand, priors 1 and 2 produce statements that are less conservative, but are also more inaccurate.  Thus if one does not have prior information regarding the similarity of true parameters, the recommended option is to use priors that lead to more conservative statements. If one indeed has prior information that the true parameters are relatively dissimilar, more informative priors (like prior 1) should be used as they lead to less sparse statements.  We briefly mention here that the same conclusions where drawn when generating data that contained lineups that were not observed in the NBA data.  Hypothetical lineups are clearly of interest to general managers and coaches as they work through the process of team building. Details are provided in Section S1 of the supplementary material.

To end this section, we provide the computational times required to optimize the global statements in \cref{fig:comp_times.results}. 
While parts of the algorithm for finding the optimal local and global statements can be parallelised to speed things up, we do not incorporate this into our code. Even without doing this, it takes less than two minutes to approximate the optimal global statement. 

\begin{figure}
\centering
\includegraphics[scale=0.45, trim={0 0cm 4.5cm 0cm},clip=true,page=7]{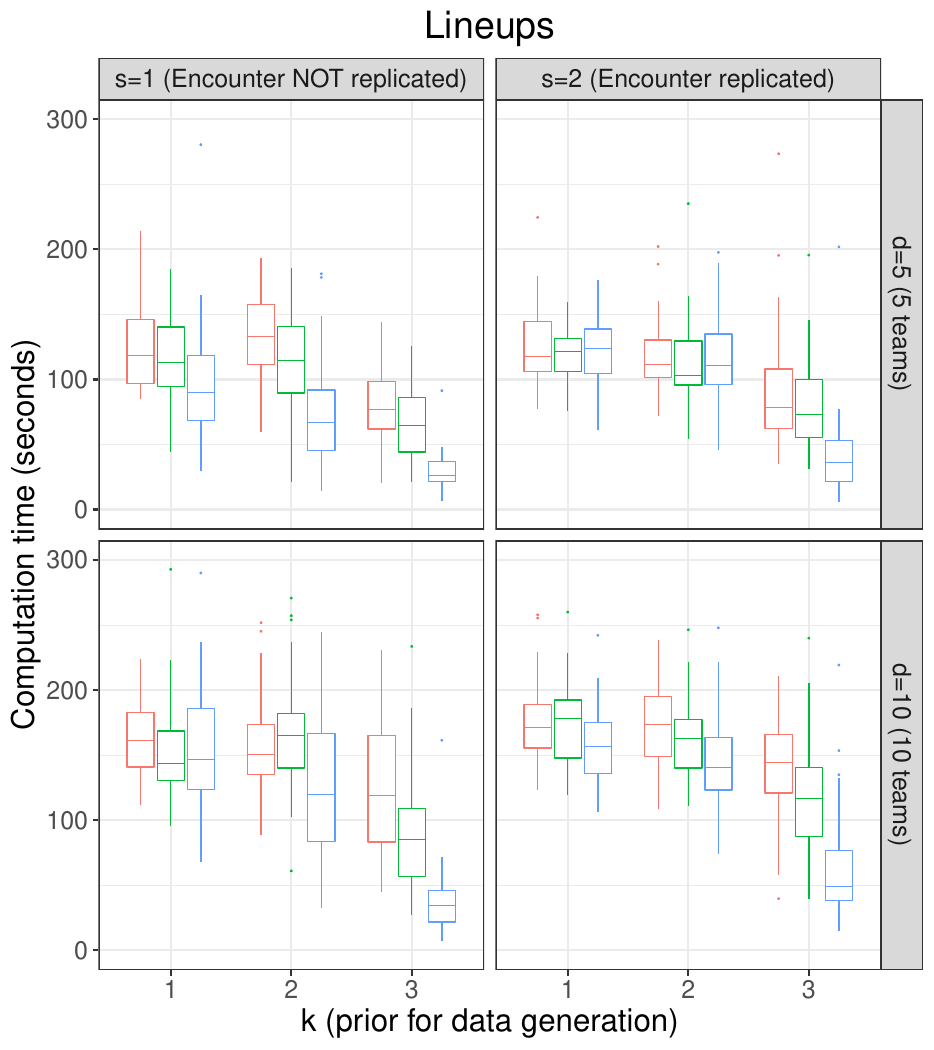}
\includegraphics[scale=0.45, trim={0.7cm 0cm 0cm 0cm},clip=true,page=7]{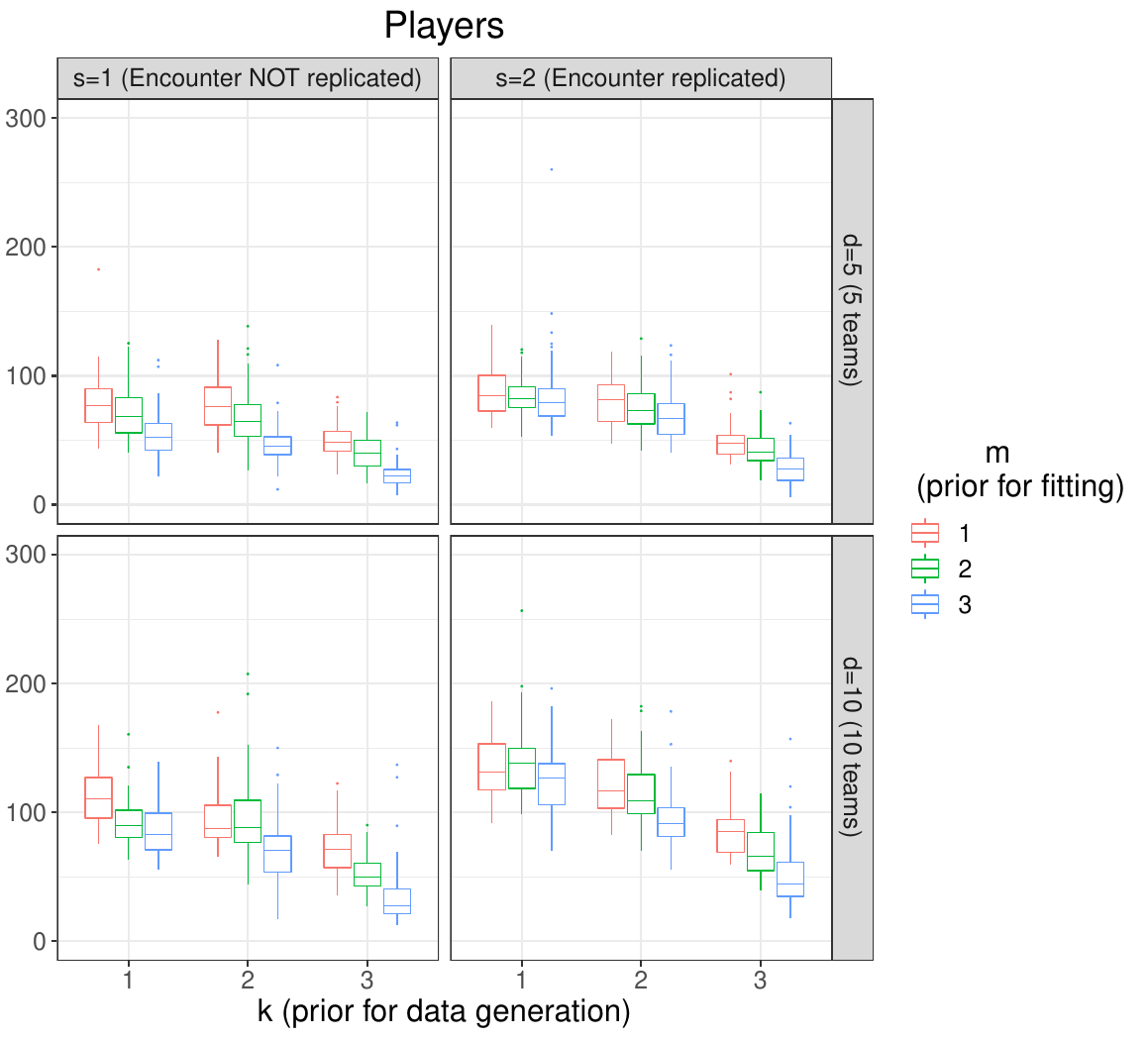}
\caption{Wall clock times required to approximate optimal global statements without parallelisation.}
\label{fig:comp_times.results}
\end{figure}

\section{Application: National Basketball Association lineup data} \label{sec.application}

The 2009-2010 NBA season consisted of $N$ = 30,807 total encounters with 10,942 unique lineups based on 30 teams and 510 players. For each encounter, the points scored by each lineup is recorded along with the five players that constituted each lineup.  For simplicity,  we treat players that change teams during the season as different players; such changes mid-season affect a small number of players. We fit the model described in \cref{sec:data.model} using prior 3. This prior was used as the simulation study indicated that it should be used in the absence of knowledge regarding the similarity of true parameters. Similarly to \cref{sec:simulationsetting}, posterior sampling was carried out using NUTS implemented in Stan.  We generated 13 independent Markov chain Monte Carlo (MCMC) chains.  The initial 2,000 iterates for each chain was discarded and each chain thinned by 20 until a total of 10,000 combined posterior draws were collected.  Then our method was applied to player and lineup abilities based on the posterior distributions of the player abilities.

Instead of optimizing over the set of all possible global statements, in applications one may want to impose additional constraints in order to obtain global statements with desired characteristics. Our framework allows this by simply letting us fix certain parameters while optimizing over the others. Setting the parameter $t$ to be very small leads to many narrow local statements (that is, lineup/player statements that contain few competing lineups/players). This can be seen from the form of the function \eqref{eq.cost}. A similar effect occurs at the global level if $q$ is very small. To illustrate the impact that $t$ and $q$ have on the resulting local and global statements,  we  consider $t = q  \in \{0, 0.1, 0.25\}$.   Unless otherwise stated, we use a grid  of 21 equally spaced values between $0$ and $0.05$ for $\alpha$ and $\gamma$ to obtain global statements; our simulations have indicated that the results do not change by using a finer grid.

\subsection{Ordering statements for lineups} \label{sec:results_lineups}

We detail results for lineups in this section.  For ease of illustration, we only consider lineups from teams that made the playoffs during the season under consideration. This subset of the data constitutes 16 teams, 261 players, 5,121 lineups, and 15,884 encounters.   

\begin{figure}
\centering
\includegraphics[scale=0.0825, trim={0 0cm 0cm 0cm},clip=true,page=3]{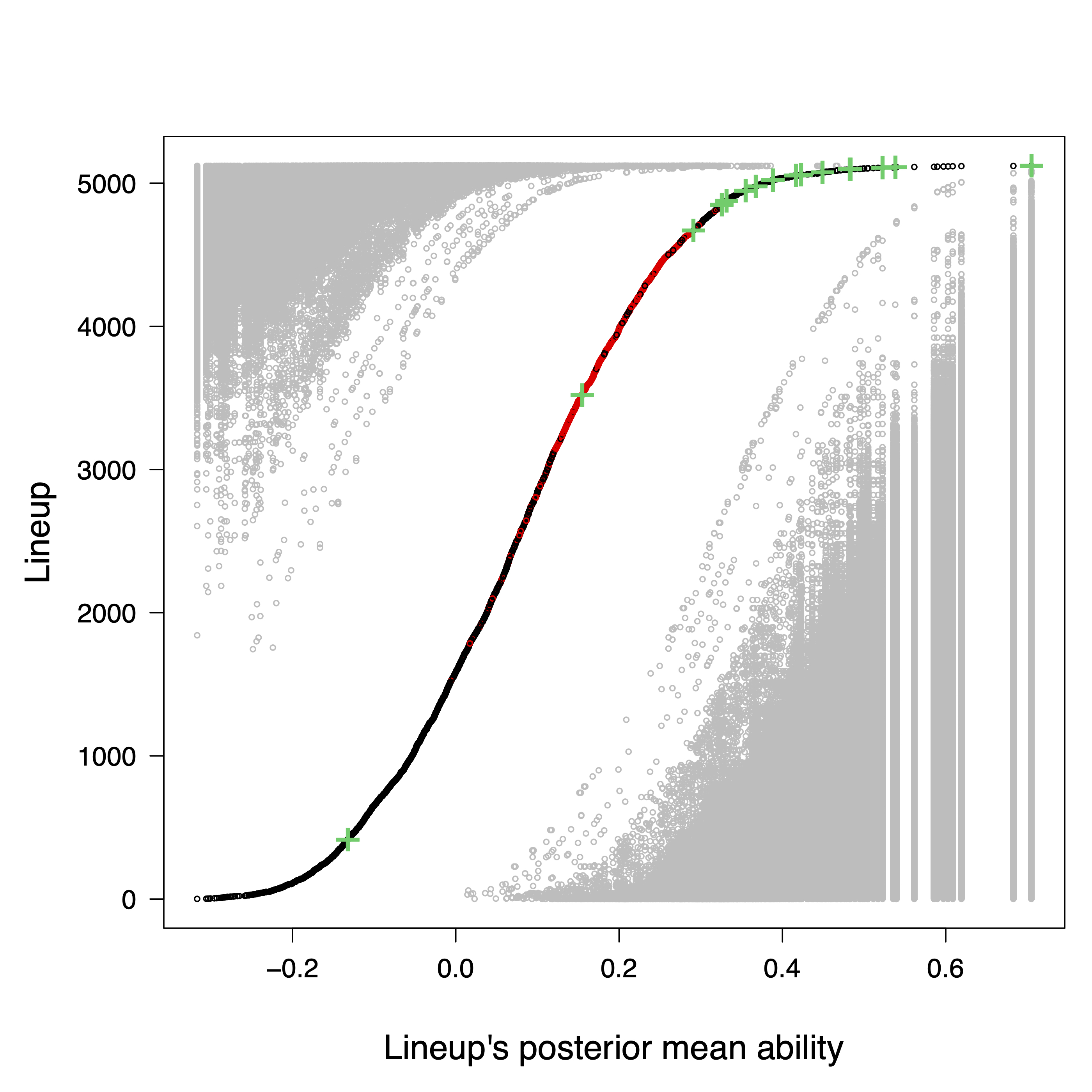}
\includegraphics[scale=0.0825, trim={0.7cm 0cm 0cm 0cm},clip=true,page=3]{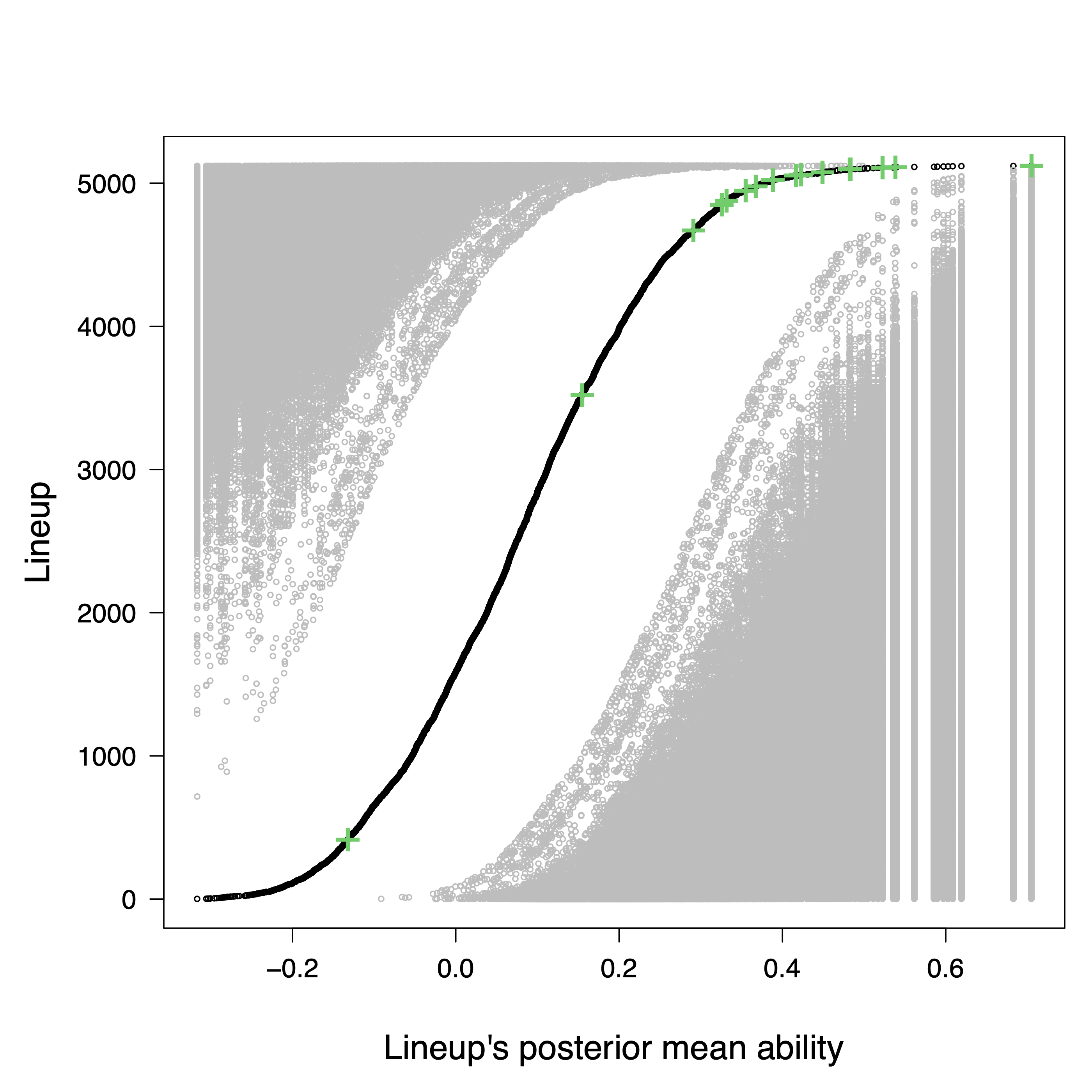}
\caption{Visual depiction of global statement based on the 5,121 lineups from the 16 playoff teams. Left and right panels display global statements setting  $(\alpha,t, \gamma, q)$ equal to $(0.025,0.1,0.05,0.1)$ and $(0.05,0.25,0.05,0.25)$, respectively. Each graph should be interpreted vertically, where each column of dots represents a particular lineup and its corresponding local statement. Red dots represent lineups to which no comparisons with any other lineups can be made. Gray dots below the black one define lineups that are worse than the corresponding lineups and those above define the group of lineups that are better. Green ``+" points correspond to the lineups listed in Table \ref{lineup_results_2}.
}
\label{fig:all_playoff_teams}
\end{figure}

\cref{fig:all_playoff_teams} displays the global statement for the 5,121 lineups based on $t=q=0.1$ and $t=q=0.25$. It took about eighteen minutes to compute this.  The posterior probability of the global statement using $t=q=0.1$ is $0.95$ and using $t=q=0.25$ is $1.0$. As expected, each lineup's local statement for $t=q=0.1$ is sparser in the sense that they contain less competing lineups compared to lineup local statements for $t=q=0.25$.  To illustrate this the red points in the left plot of \cref{fig:all_playoff_teams} correspond to lineups whose local statements are empty.  That is, these lineups are no better or worse than any other lineup.  On the other hand, all lineup local statements based on $t=0.25$ are non-empty.  This added ``information'' comes at the cost of accuracy as more errors are permitted for each lineup.  We therefore recommend exploring global statements for a range of values for $t$ and $q$, and making decisions depending on the amount of risk that a decision maker is willing to take. 

The lineup that was better than the greatest number of other lineups is from the Orlando Magic and is comprised of Matt Barnes, Vince Carter, Dwight Howard, Rashard Lewis, and Jameer Nelson.  This lineup was better than 4,803 of the  5,121 lineups and worse than none in their local statement for $t=q=0.1$, and they were better than 5,019 lineups and worse than none in their local statement for $t=q=0.25$. Table  \ref{lineup_results_2} contains the ``best'' lineups (that is, those whose local statement contains the highest number of competing lineups for which they are better) extracted from the global statement for each of the playoff teams under both $t=q=0.25$ and $t=q=0.1$. Interestingly, moving from $t=0.1$ to $t=0.25$ results in a larger increase of lineups ``below'' than lineups ``above'' for each team.  Each of the lineups could be considered to be their teams ``best''.  To visualize where each lineup in Table \ref{lineup_results_2} is located in the global statement, they are highlighted by green colored ``+'' points in \cref{fig:all_playoff_teams}.  It should be clear why, for example, Chicago's lineup is not better than any other lineup, as it is found in the lower tail of lineup abilities.  It may be tempting at first glance to use the information like that provided in Table \ref{lineup_results_2} to rank order lineups.  Doing so however, is problematic for two reasons.  First, it is not clear how to rank lineups that have empty local statements (i.e., those that are no better or worse then any other lineup).  Second, it is possible that there exist two lineups whose local statements are the same size but are comprised of different lineups.  Here too, it is unclear how to rank the two lineups.

\begin{table}[ht]
\caption{Results when comparing lineups from all playoff teams. 
Each team's listed lineup corresponds to that which was better than the highest number of lineups among all playoff teams.  The columns  ``below'' and ``above'' refer to the number of lineups worse and better, respectively, than the stated lineup across all playoff teams.}
\centering
\begin{tabular}{cl rr rr}
\toprule
& & \multicolumn{2}{c}{$t=q=0.1$} & \multicolumn{2}{c}{$t=q=0.25$} \\\cmidrule(lr){3-4} \cmidrule(lr){5-6}
Team & Lineup & Below & Above & Below & Above \\ 
\midrule
ATL & Bibby, Horford, Johnson, Josh Smith, Williams &      2241 &   0 & 3230 &   0 \\
BOS & Garnett, Perkins, Pierce, R.Allen, Rondo &           1696 &   0 & 2624 &   1 \\ 
CHA & Augustin, Diaw, Jackson, Murray, Wallace &           3 &   9 &  51 &  26 \\ 
CHI & Alexander, D.Brown, Law, Pargo, Richard &            0 &  89 &   0 & 332 \\ 
CLE & James, Jamison, Parker, Varejao, West &              2505 &   0 & 3473 &   0 \\ 
DAL & Dampier, Kidd, Marion, Nowitzki, Terry &             1428 &   0 & 2347 &   0 \\
DEN & Andersen, Anthony, Billups, Nene, Smith &            555 &   1 & 1205 &   2 \\ 
LAL & Artest, Bryant, Bynum, Fisher, Gasol &               2940 &   0 & 3818 &   0 \\ 
MIA & Alston, Haslem, ONeal, Richardson, Wade &            639 &   0 & 1370 &   1 \\ 
MIL & Bogut, Delfino, Ilyasova, Jennings, Salmons &        168 &   1 & 645 &   2 \\ 
OKC & Collison, Durant, Krstic, Sefolosha, Westbrook &     890 &   0 & 1723 &   1 \\ 
ORL & Barnes, Carter, Howard, Lewis, Nelson &              4803 &   0 & 5019 &   0 \\ 
PHX & Dudley, Frye, Nash, Richardson, Stoudemire &         948 &   1 & 1774 &   2 \\ 
POR & Aldridge, Batum, Camby, Miller, Roy &                1410 &   0 & 2418 &   0 \\ 
SAS & Duncan, Ginobili, Jefferson, McDyess, Parker &       2168 &   0 & 3184 &   0 \\ 
UTA & Boozer, Kirilenko, Korver, Millsap, Williams &       452 &   1 & 1094 &   2 \\ 
\bottomrule
\end{tabular}
\label{lineup_results_2}
\end{table}

It is somewhat challenging to summarize all the information found in \cref{fig:all_playoff_teams}. From a coaching perspective, focusing specifically on their team's lineups would be more useful as this would allow coaches to more easily identify their team's best lineups and game plan for a particular opponent.  To illustrate this, we extract the 305 local lineup statements from the global statement for the Utah Jazz and display them in \cref{fig:utah_jazz}. We note that all local statements are guaranteed to have posterior probability at least that of the global statement. As before, setting $t=q=0.1$ results in local statements that are sparser relative to those for $t=q=0.25$, and the red points found in the left plot of \cref{fig:utah_jazz} correspond to lineups whose local statements are empty.

Unsurprisingly, for both $t=q=0.1$ and $t=q=0.25$, Deron Williams was a member of all the ``best'' Utah Jazz lineups.  Perhaps somewhat surprisingly, Kyle Korver rather than Carlos Boozer (the supposed 2nd best player on the Jazz) was featured more regularly in the ``better'' Jazz lineups.  Very suprisingly, Kyrylo Fesenko (a seldom used player) was a regular member of the Jazz's ``better'' lineups.  In fact the lineup with the largest estimated ability is comprised of Carlos Boozer, Kyrylo Fesenko, Andrei Kirilenko, Kyle Korver, and Deron Williams. This lineup played sparingly (a total of five encounters), but has a local statement that is not empty (better than 305 other lineups most of which didn't contain Deron Williams). Thus, it seems that the Utah Jazz coach should have explored lineups that contained Fesenko more.  Contrast this to the lineup of Andrei Kirilenko, Kyle Korver, Wesley Matthews, C.J. Miles and Deron Williams which played more regularly and whose estimated ability is similar (the red dot with the largest estimated ability in left plot of Figure \cref{fig:utah_jazz}).  However, the local statement associated with this lineup is empty.  The reason behind this is the uncertainty associated with the lineup ability estimates.  The posterior standard deviation of the latter is larger than the former. Remaining silent about this lineup highlights a novel feature of our approach.  It turns out that the Utah Jazz lineup with the local statement containing the highest number of competing lineups consists of Carlos Boozer, Andrei Kirilenko, Kyle Korver, Paul Millsap, and Deron Williams. This lineup was arguably the strongest for the Utah Jazz during the 2009-2010 season.

\begin{figure}
\centering
\includegraphics[scale=0.4, trim={0 0cm 0cm 0cm},clip=true]{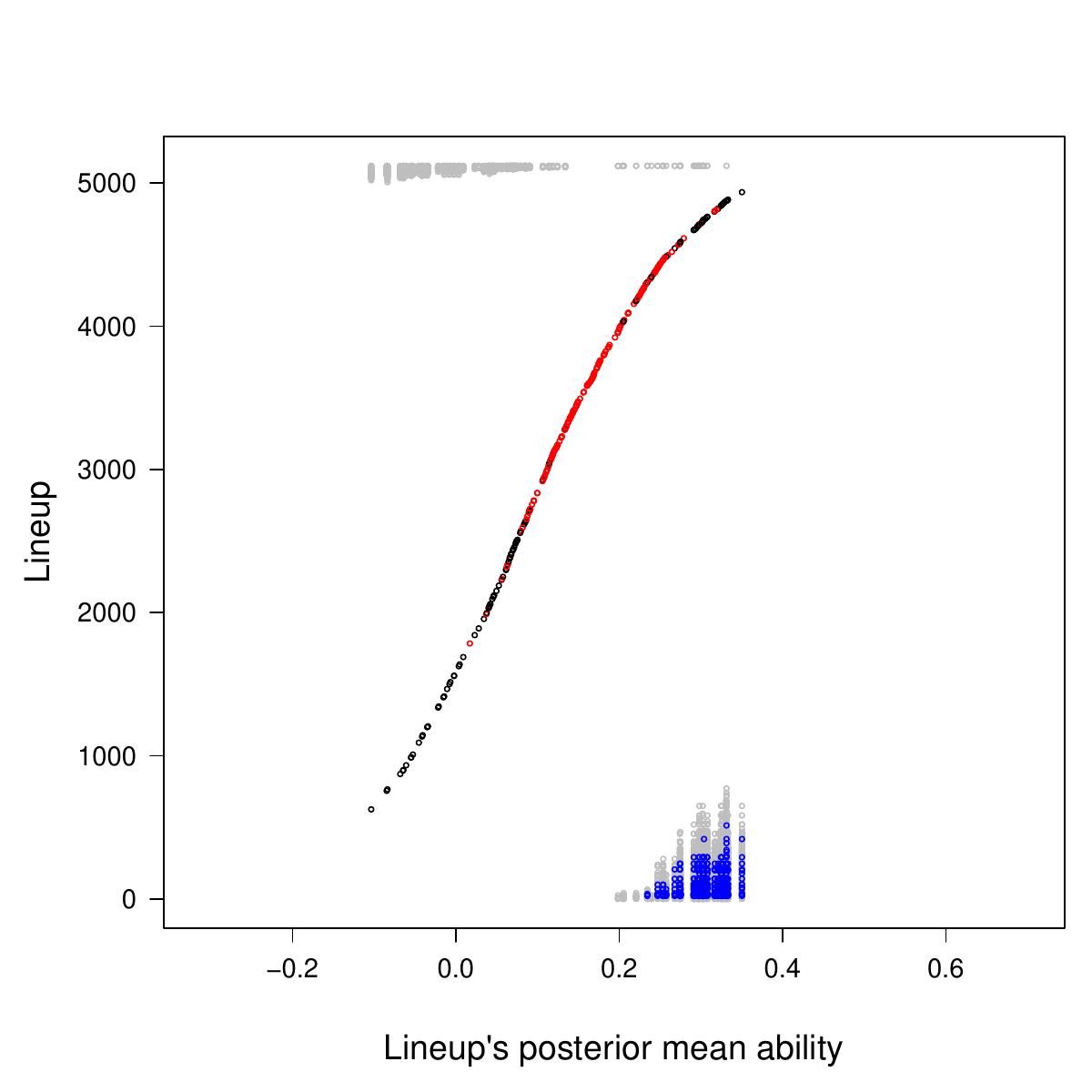}
\includegraphics[scale=0.4, trim={0.4cm 0cm 0cm 0cm},clip=true]{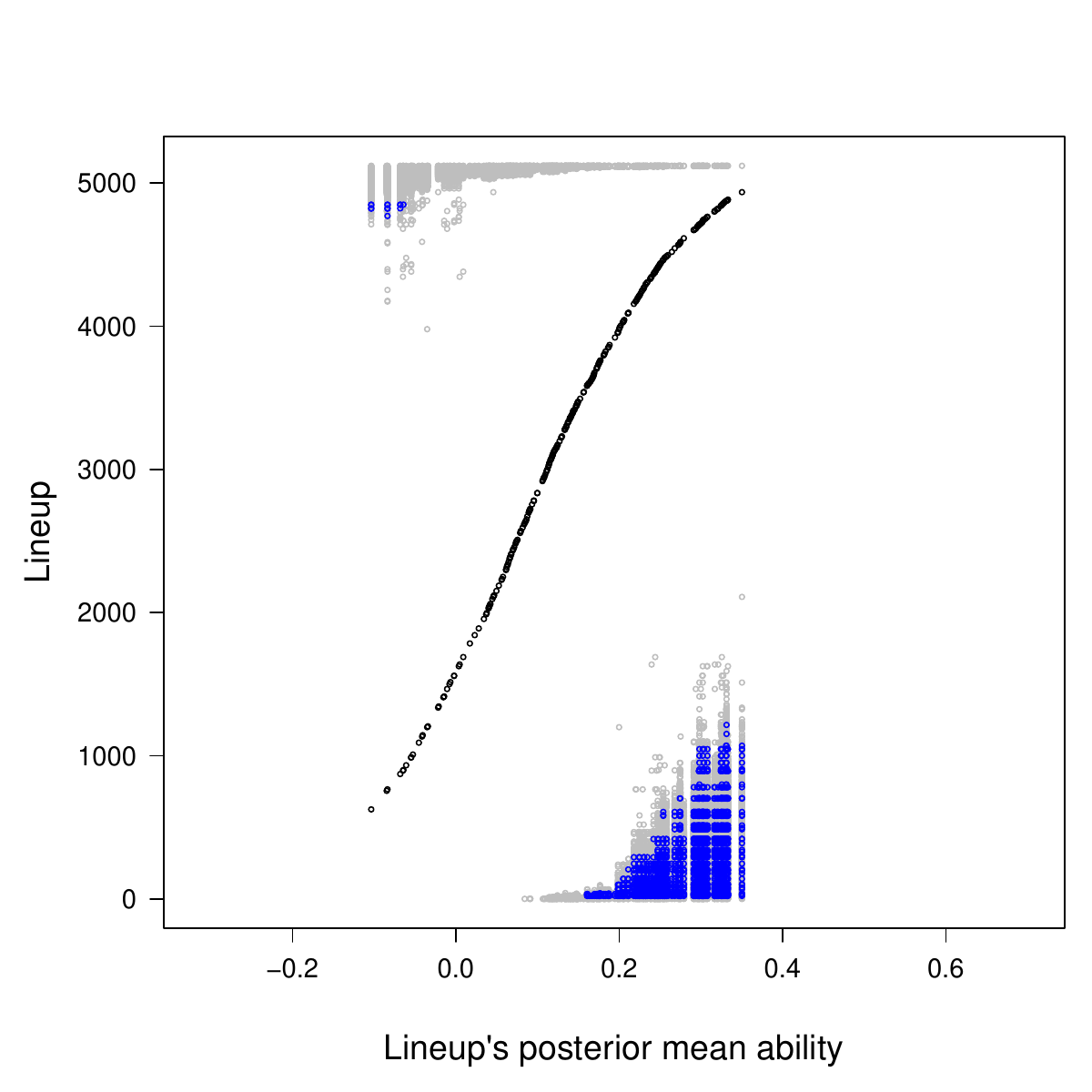}
\caption{Graphical display of the 301 lineups that the Utah Jazz employed during the 2009-2010 NBA season. Blue dots correspond to lineups from the Denver Nuggets. Left and right panels display global statements setting  $(\alpha,t, \gamma, q)$ equal to $(0.025,0.1,0.05,0.1)$ and $(0.05,0.25,0.05,0.25)$, respectively. Each graph should be interpreted vertically, where each column of dots represents a particular lineup and its corresponding local statement. Red dots represent lineups to which no comparisons with any other lineups can be made. Gray dots below the black one define lineups that are worse than the corresponding lineups and those above define the group of lineups that are better.}
\label{fig:utah_jazz}
\end{figure}

It is valuable for a coach to know which of their team's lineups should compete against particular lineups of the opposing team.  In the first round of the playoffs of the 2009-2010 NBA season, the Utah Jazz played against the Denver Nuggets. Each blue point in \cref{fig:utah_jazz} corresponds to a particular lineup from the Denver Nuggets.  Thus we see that the ``best'' Jazz lineups are better than a number of Denver Nugget lineups for $t=q=0.1$.  However, it is only when we allow more uncertainty regarding the accuracy of the local lineups that we see some Denver Nugget lineups that are better than any Utah Jazz lineups.

\subsection{Ordering statements for players}\label{sec:results_players}

Ideally, global statements would consist of dense local statements that have high posterior probability and small values for $\alpha$, $\gamma$, $t$ and $q$. Unfortunately, this turns out to be impossible to achieve when comparing abilities of all 510 players.  Figure \ref{fig:all_players_teamsS2} of the Supplementary Material displays the global statement of all 510 players for $t=q=0.1$ and $t=q=0.25$.  It took under a minute to compute these global statements.  The posterior probability is above $0.95$ in both cases, but local statements are very sparse, thus limiting the usefulness of the global statement. This is a consequence of the uncertainty associated with player ability estimates (see Figure \ref{fig:all_players_teamsS1} of the online Supplementary Material) and producing a global statement based on all 510 players. Each additional player included in the global statement introduces more uncertainty which propagates to the local statements.  

That being said, comparing all 510 players is typically not of interest as players can take on vastly different roles within a team framework.  To further illustrate our method, we apply our procedure to a subset of players whose usage rate (an estimate of the percentage of a team's possessions ended by a particular player's action) was in the top twenty among all players in the 2009-2010 season.  This permits us to compare players with similar roles across teams.   As a brief aside, we remark that player abilities are estimated borrowing information across players that play on the same team.  As a result, players who played on bad teams (for example, Golden State Warriors) were often members of lineups that lost many encounters (that is, negative difference between points scored and points allowed), and their ability was estimated to be generally lower than players on strong teams.  Thus,  small contributions to winning lineups result in a higher estimated ability relative to larger contributions to losing lineups.

The posterior distributions for players' abilities whose usage rate belonged to the top twenty during the 2009-2010 season are provided in the top-left plot of \cref{fig:usage_rateALL}.    For this subset of players, we fix $t=q=0$ as this will permit us to use the local statements to produce multiple orderings that can all be interpreted similarly to traditional rankings.  Note that unlike traditional rankings, the probability that all these rankings hold simultaneously is equal to the posterior probability of the global statement. Further,  the probability of any subset of these rankings holding simultaneously is lower bounded by the posterior probability of the global statement. To rank players using the global statement, we only need to look at the inherent transitivity of the local statements. That is, if player 1 is better than players 2 and 3, and player 3 better than player 2, then we can conclude that player 1 is better than player 3 is better than player 2 (i.e., $\xi_2<\xi_3<\xi_1$). Although the transitivity of the local statements can be used to create rankings of subsets of players when either $t \neq 0$ or $q \neq 0$ (that is, local and global errors are allowed), we refrained from exploring them as it is not evident how to quantify the inherited error for each of the derived rankings.

The global statement in the top-right plot of \cref{fig:usage_rateALL} has a posterior probability of $0.87$.  This statement is sparse and only produces rankings that involve at most two players (for example, $\xi_{\rm  James} >\xi_{\rm Ellis}$ and $\xi_{\rm Wade} > \xi_{\rm Ellis}$ , et cetera.).  Even though it is sparse, we can still conclude that the players with empty local statements (e.g., players with red dots ``POR\_Roy'' - ``CHI\_Rose'') are essentially exchangeable in terms of the their abilities.  However, their salaries are very different ranging from 1 million to 16 million USD a year.  Of course some of the discrepancy in salary is due to experience in the league.  But even players like Stephen Jackson and J. R. Smith who play very similar roles, using a similar style of play, and with a similar number of years in the league, differ in salary by more than 2 million USD.  Thus, our approach can provide general managers added information when making personnel decisions by indicating which players are very similar in their contributions to lineup success.  The global statement provided in the bottom left of \cref{fig:usage_rateALL} is less sparse with a posterior probability of $0.34$ and, given that this probability is low, we provide it more to contrast the sparser global statement.  In this global statement, Lebron James and Dwayne Wade were better than 10 and 9 players, respectively, and worse than none.  Monta Ellis, on the other hand, was worse than 11 players and better than none.     

To help visualize the relationships between players, the global statement can be plotted using a network graph. The bottom right plot of \cref{fig:usage_rateALL} displays a network graph that corresponds to the global statement in the bottom-left plot. Each node corresponds to a player. The network graph is directed in that edges connecting two players indicate which player is better by way of an arrow. Thus, any path of the graph represents a ranking of the nodes (that is, players).  From the network graph, its clear that Monta Ellis is ranked as the worst player among the twenty, as many arrows are pointing away from him. Also notice that Monta Ellis is ranked below Carmelo Anthony, who is ranked below Lebron James. So, we can conclude that $\xi_{\rm James} > \xi_{\rm Anthony} > \xi_{\rm Ellis}$.   Similarly, we can  conclude that $\xi_{\rm Durant} > \xi_{\rm Rose}$ and $\xi_{\rm Bryant} > \xi_{\rm Hamilton}$, but also that Bryant is neither better nor worse than James, et cetera. Note that all these ranking statements, or any combination of rankings from the global statement, have a joint posterior probability that is bounded below by $0.34$ (which is the posterior probability of the global statement).

Twenty players are displayed in the top left plot of \cref{fig:usage_rateALL}, where the order of the players on the $x$-axis corresponds to their median ability. Given that no ranking of the twenty players appears more than once among the $10^4$ MCMC iterates, any point estimate of the ranking will have posterior probability close to zero.  Even if focus is placed on rankings that only require James, Wade, and Durant to be in the top three regardless of order, the posterior probability is still only $9 \times 10^{-2}$. Similarly, the  posterior probability  of rankings where Ellis, Maggette, and Kaman are in the bottom three regardless of order is $10^{-2}$, and the posterior probability of rankings where both  James, Wade, and Durant are in the top three and  Ellis, Maggette, and Kaman are in the bottom three simultaneously is $7 \times 10^{-4}$.  In contrast, our global statement communicates a similar message about these six players but with a much higher probability ($3.4 \times 10^{-1}$ compared to $2.5 \times 10^{-3}$). Notice that our global statement indicates that, for example, James, Wade, and Durant are the players who have the largest number of players that have a lower ability (see left bottom plot of \cref{fig:usage_rateALL}). Further, the global statement is saying that these three players are better than the seven players below Smith, and not that they have an ability that is consistently greater than the ability of all other 17 players.

\begin{figure}
\centering
\includegraphics[scale=0.4, trim={0 0cm 0cm 0cm},clip=true,page=1]{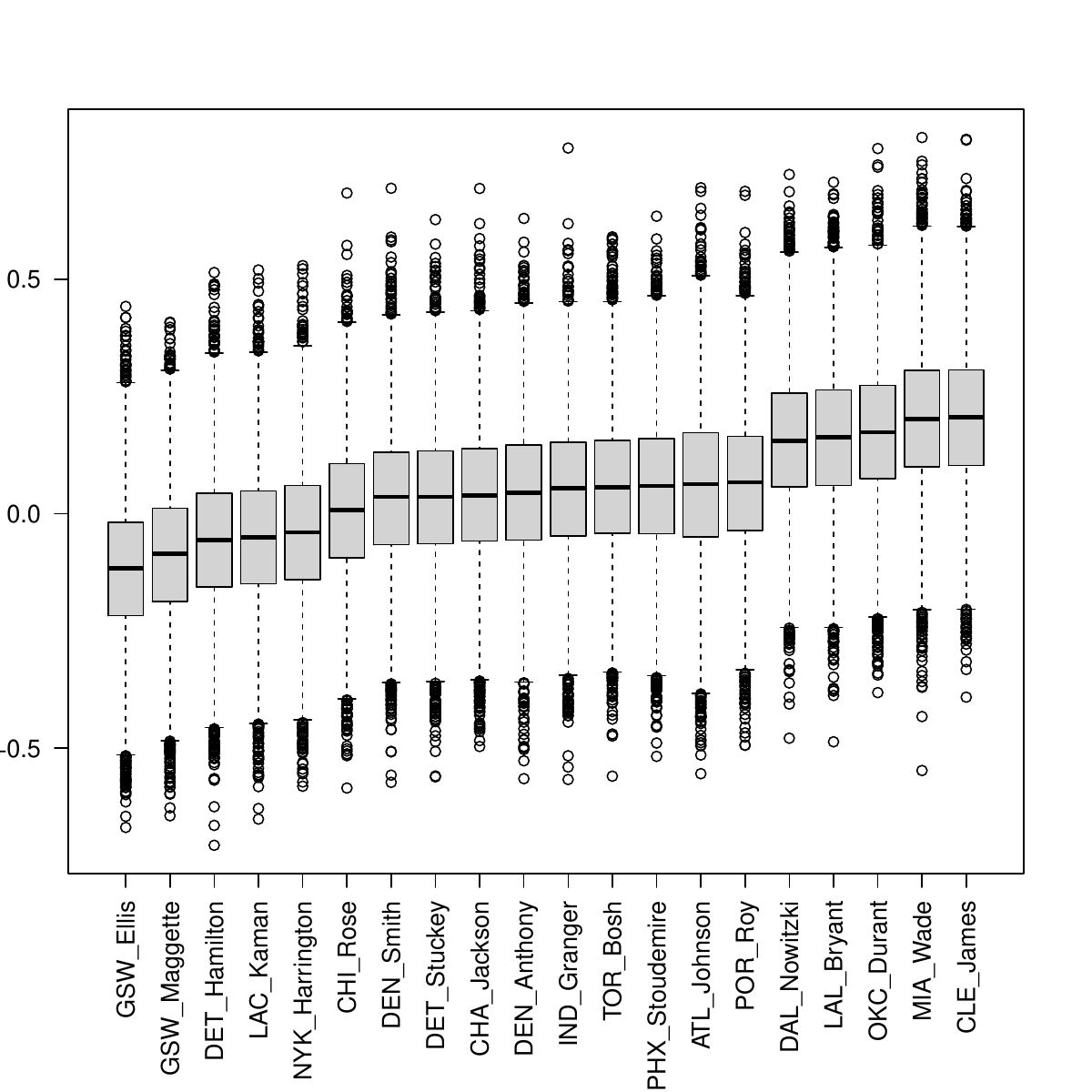}
\includegraphics[scale=0.4, trim={0 0cm 0cm 0cm},clip=true,page=2]{plots/usageRateALL4.pdf}
\includegraphics[scale=0.4, trim={0 0cm 0cm 0cm},clip=true,page=3]{plots/usageRateALL4.pdf}
\includegraphics[scale=0.8, trim={5mm -5mm 6mm 6mm},clip=true]{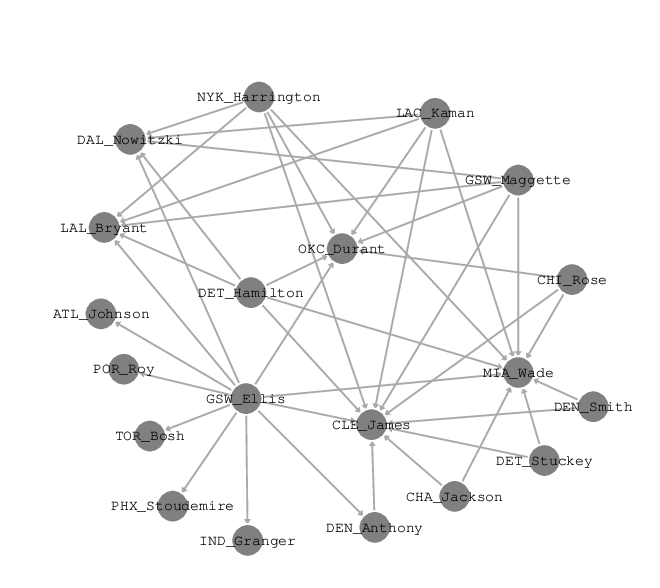}
\caption{
Results from players that were in the top twenty in terms of usage rate. The top-left plot corresponds to the posterior distributions of the player abilities for each of the twenty players. The top-right and bottom-left plots display the global statements setting  $(\alpha,t, \gamma, q)$ equal to $(0.0225,0,0.05,0)$ and $(0.09,0,0.9,0)$, respectively. These two plots can be interpreted similarly as \cref{fig:all_playoff_teams}.  The bottom-right plot displays the global statement in the bottom left-plot using a network graph. Arrows point to nodes (players) that are better than those represented by the respective nodes.
}
\label{fig:usage_rateALL}
\end{figure}

\section{Discussion}\label{sec:discussion}

We have presented a novel procedure to find statements associated with orderings of parameters (for example, player and lineup abilities). Relying on a decision theoretic framework, the proposed approach generates statements associated with high posterior probability while incorporating as many players/lineups into the statements as possible. This represents a new approach that is quite different from existing ranking methods.

Although we use a prior distribution that assigns zero probability to ties among parameters in the numerical illustrations, we can incorporate ties in the statements by using prior specifications that assign positive probability to them, (for example, spike-and-slab priors of \citealp{ishwaran2005spike}), or by defining ties based on practical criteria (for example, a tie is present if the maximum difference between parameters is below a given threshold). Moreover, it would be interesting to consider informative priors for player abilities. A possibility is to modify the prior mean for $\xi_l$ in \cref{sec:simulationsetting} to $\mu_l$, with $\mu_l$ depending on the official NBA ranking of the players in the previous year (that is, have a different prior mean for each player ability).

From an NBA general manager's perspective, it may seem unsatisfactory that the global statements are silent about some players.  This could be particularly true when employing the lineup data to make personnel decisions (for example, deciding between which free-agents to pursue). In this setting, our procedure can focus specifically on the small subset of players of interest.  Posterior probabilities of the form $\PP(\xi_l > \xi_{l'} \mid \data)$ can be useful for deciding whether to trade players. Since there is often a monetary cost involved in trading, this can be incorporated as well to calculate the expected gain in trading players. 

\section*{Acknowledgement}\label{sec.Acknowledgements}

This work was supported by grant W911NF-16-1-0544 of the U.S. Army Research Institute for the Behavioral and Social Sciences (ARI).  The third author gratefully acknowledges support from the Basque Government through the BERC 2018-2021 program, by the Spanish Ministry of Science, Innovation and Universities through BCAM Severo Ochoa accreditation SEV-2017-0718.

\appendix 
\section{Algorithm for local and global statements}\label{app:algorithm}


From an empirical perspective, assume we have samples $\bxi^{(1)}, \dots, \bxi^{(M)} \sim \PP$, which can be the output from an Markov Chain Monte Carlo algorithm, where $\bxi^{(i)} = (\xi_1^{(i)}, \dots, \xi_{\ntotal}^{(i)})$, $i=1, \dots, M$.
The sample versions of the local sets $\overline{\A}_{\pidx,\alpha}$ and $\underline{\A}_{\pidx,\alpha}$ are 
\begin{align}  \label{eq.Adef_sample}
\begin{aligned}
\overline{\A}_{\pidx,\alpha}^{\mathrm{sample}}
& = 
\left \{ \pidx' \, : \, \frac1M \sum_{i=1}^M \indicator ( \xi^{(i)}_{\pidx'}>\xi^{(i)}_{\pidx} ) > 1 - \alpha  \right \}, 
\\
\underline{\A}_{\pidx,\alpha}^{\mathrm{sample}}
& =
\left \{ \pidx' \, : \, \frac1M \sum_{i=1}^M \indicator  ( \xi^{(i)}_{\pidx}>\xi^{(i)}_{\pidx'}  ) > 1 - \alpha  \right \};
\end{aligned}
\end{align}
finding these can be trivially parallelised for each $\pidx = 1, \dots, \ntotal$. 
Using \cref{eq.Adef_sample}, we verify whether the local statement \eqref{eq.Adef2} holds at each MCMC iteration by  
\begin{align} \label{eq.Adef2_sample} 
\indicator(A_{\pidx,\alpha, t}^{(i), \, \mathrm{sample}})
& = 
\indicator\left ( 
\frac{ \sum_{\pidx' \in \overline{A}_{\pidx, \alpha}} \indicator (\xi_{\pidx'}^{(i)} > \xi_\pidx^{(i)})  + \sum_{\pidx' \in \underline{A}_{\pidx, \alpha}} \indicator ( \xi_{\pidx}^{(i)} > \xi_{\pidx'}^{(i)} ) 
}{|\underline{\A}_{\pidx,\alpha} \cup \overline{\A}_{\pidx,\alpha}|} > 1-t
\right ) \in \{0,1\},
\end{align}
and finding this is also trivially parallelisable for each $i = 1, \dots, M$. The sample version of the set \eqref{eq.Gdef} is 
\begin{align} \label{eq.Gdef_sample}
\G_{\alpha, t, \gamma}^{\mathrm{sample}}
& = 
\left \{ 
l \, : \, \frac1M \sum_{i=1}^M \indicator ( A_{\pidx,\alpha, t}^{(i), \, \mathrm{sample}} )
\geq 
1-\gamma \right \},
\end{align}
which is again parallelisable as each term in the sum can be computed independently of each other. For any $q \in [0,1]$, this leads to a global statement given by \cref{eq.global_statement}. The probability of this global statement being true is estimated from the samples as 
\begin{align} \label{eq.global_prob_sample}
\frac1M \sum_{i=1}^M  \indicator \left \{ \frac{ \sum_{\pidx \in \G_{\alpha, t, \gamma}^{\mathrm{sample}}} \indicator (A_{\pidx,\alpha, t}^{(i), \, \mathrm{sample}})}
{|\G_{\alpha,t,\gamma}^{\mathrm{sample}}|} > 1-q   \right \}.
\end{align}
This leads to the following \cref{alg.ranking}, which along with \cref{eq.global_prob_sample} provides a complete recipe for evaluating the global statements. 
\begin{algorithm}
\caption{Bayesian global ordering statements.} 
\label{alg.ranking}

\textbf{Input:} Samples $\bxi^{(1)}, \dots, \bxi^{(M)} \sim \PP$, input parameters $(\pidx, \alpha, t, \gamma) \in [0,1]^4$.

\begin{algorithmic}[1] 

\FOR{$\pidx = 1, \ldots, \ntotal$}

\STATE 
Compute $\overline{\A}_{\pidx,\alpha}^{\mathrm{sample}}$ and $\underline{\A}_{\pidx,\alpha}^{\mathrm{sample}}$ by \cref{eq.Adef_sample}.

\ENDFOR 

\STATE 
Compute $\G_{\alpha, t, \gamma}^{\mathrm{sample}}$ by \cref{eq.Gdef_sample}.

\STATE 
Define $G_{\alpha,t, \gamma, q}$ by \cref{eq.global_statement}. 

\end{algorithmic}

\textbf{Output:} Global statement $G_{\alpha,t, \gamma, q}$, local sets $\overline{\A}_{\pidx,\alpha}$ and $\underline{\A}_{\pidx,\alpha}$, and global set $\G_{\alpha,t, \gamma}$.
\end{algorithm}

\section{Algorithm for optimal global statement} \label{app:opt}

We describe an algorithm for approximating the optimal global statement, where the optima is as described in \cref{sec.decision} of the main text. For computational reasons, we restrict the search to a fine grid for $\alpha$, $0 = \alpha_0 < \cdots < \alpha_k = \alpha_{\max}$. 
Evaluating equations \eqref{eq.Adef_sample} and \eqref{eq.Adef2_sample} for $\alpha \in \{\alpha_0, \dots, \alpha_k\}$ are the most expensive part of \cref{alg.ranking}, and these are performed at the start. The objective function \eqref{eq.reward_fn} we seek to maximize is non-convex and its gradients cannot be calculated, and we adopt a variant of the pattern search algorithm \citep{hooke1961direct} to optimize it. 
To ensure a minimum credible level, we can restrict the search to global statements whose posterior probability \eqref{eq.global_prob_sample} is above given threshold $1-\epsilon$, for $\epsilon \in (0,1)$, which is achieved by setting $R(a)$ at zero if the posterior probability exceeds $1-\epsilon$. In addition, we restrict the search to $(\alpha,t,\gamma,q) \in [0,0.05] \times [0,0.1] \times [0,0.5] \times [0,0.1]$ in order to obtain statements with low local and global errors. 
For a point $(\alpha, t, \gamma, q)$ and step-size $\delta>0$, we use the notation $(\alpha \pm \delta, t \pm \delta, \gamma \pm \delta, q \pm \delta)$ to refer to the set of the $3^4$ coordinate-wise perturbations of size $\delta$ around $(\alpha, t, \gamma, q)$, where the perturbed point is set to the boundary of $[0,0.05] \times [0,0.1] \times [0,0.5] \times [0,0.05]$ if it crosses it.
An algorithm to find a local optima given a starting point $(\alpha^{(0)}, t^{(0)}, \gamma^{(0)}, q^{(0)})$ is provided in \cref{alg.opt}. We run \cref{alg.opt} with multiple starting points and choose the optimal global statement among these. 

\begin{algorithm}
\caption{Variant of pattern search algorithm.} 
\label{alg.opt}

\textbf{Input:} Samples $\bxi^{(1)}, \dots, \bxi^{(M)}$, grid $0 = \alpha_0 < \cdots < \alpha_k = \alpha_{\max}$, initial point $a^{(0)} = (\alpha^{(0)}, t^{(0)}, \gamma^{(0)}, q^{(0)})$ such that $\alpha^{(0)} \in \{\alpha_0, \dots, \alpha_k\}$, initial step-size $\delta$, minimum step-size $\delta_{\min}$. 

\begin{algorithmic}[1] 

\STATE 
Evaluate equations \eqref{eq.Adef_sample} and \eqref{eq.Adef2_sample} for $\alpha \in \{\alpha_0, \dots, \alpha_k\}$.

\FOR{$s=1,2, \dots$}

\STATE 
Evaluate objective function at points $(\alpha^{(s-1)} \pm \delta, t^{(s-1)} \pm \delta, \gamma^{(s-1)} \pm \delta, q^{(s-1)} \pm \delta)$.

\STATE 
Set $a^{(s)} = (\alpha^{(s)}, t^{(s)}, \gamma^{(s)}, q^{(s)}) = \argmax_{a \in (\alpha^{(s-1)} \pm \delta, \, t^{(s-1)} \pm \delta, \, \gamma^{(s-1)} \pm \delta, \, q^{(s-1)} \pm \delta)} R(a)$. 

\IF{$R(a^{(s)}) = R(a^{(s-1)})$}

\IF{$\delta/2 < \delta_{\min}$}

\STATE 
End algorithm.

\ELSE 

\STATE 
Set $\delta = \delta/2$.

\STATE 
Go to step 3.  

\ENDIF 
\ENDIF 

\ENDFOR 

\end{algorithmic}

\textbf{Output:} Optimal action $a^{(s)}$and associated global statement $G_{\alpha^{(s)}, \, t^{(s)}, \, \gamma^{(s)}, \, q^{(s)}}$.
\end{algorithm}

\bibliographystyle{asa}
\bibliography{reference}

\makeatletter 
\renewcommand{\thefigure}{S\@arabic\c@figure}
\renewcommand{\thesection}{S\@arabic\c@section}
\makeatother

\newpage 
\setcounter{section}{0}
\setcounter{figure}{0}

\begin{center}
\textbf{\Large Supplementary Material for ``Bayesian inferences on uncertain ranks and orderings: Application to ranking players and lineups''}
\end{center}

This document contains supplementary material for the article ``Bayesian inferences on uncertain ranks and orderings: Application to ranking players and lineups''. 
Section \ref{supp.figs} contains simulation results for estimated player and lineup abilities as is discussed in Section 4.3 of the main text. 
Section \ref{supp.player.comparison} contains the posterior abilities of all the 510 players in the 2009-10 NBA season as well as the global statement for all the players for $t=q=0.1$ and $t=q=0.25$.

\section{Complementary figures for the simulation study} \label{supp.figs}

The results obtained for players and lineups for the 50 simulated datasets share similar patterns.  For this reason,  results associated with posterior probabilities, frequentist coverage, and average number of elements in the global statements for lineups  are presented in the main manuscript, while those for players are presented in \cref{fig:simulation.resultsS1} of this section.  Local and global errors for statements associated with players and lineups are both reported in \cref{fig:simulation.resultsS2}. Notice that, when using a prior that promotes shrinkage towards the mean (such as prior 3), the global statements are more sparse. Such sparsity leads to more conservative statements which in turns improves the frequentist coverage without having too much effect on the posterior probability of the statement, see \cref{fig:simulation.resultsS1}.  \cref{fig:simulation.resultsS2} shows how sparsity is closely related to the local and global errors, i.e., the lower the errors the more conservative the statements are.  Figures \ref{fig:simulation.resultsS1} and \cref{fig:simulation.resultsS2} also shows that, to find less sparse statements, reduce the errors, and improve the frequentest coverage, we will need more accurate ability estimates which can be achieved by collecting more data (i.e., moving from $s=1$ to $s=2$).

\begin{figure}[H]
\centering
\includegraphics[scale=0.45, trim={0 0cm 4.5cm 0cm},clip=true,page=1]{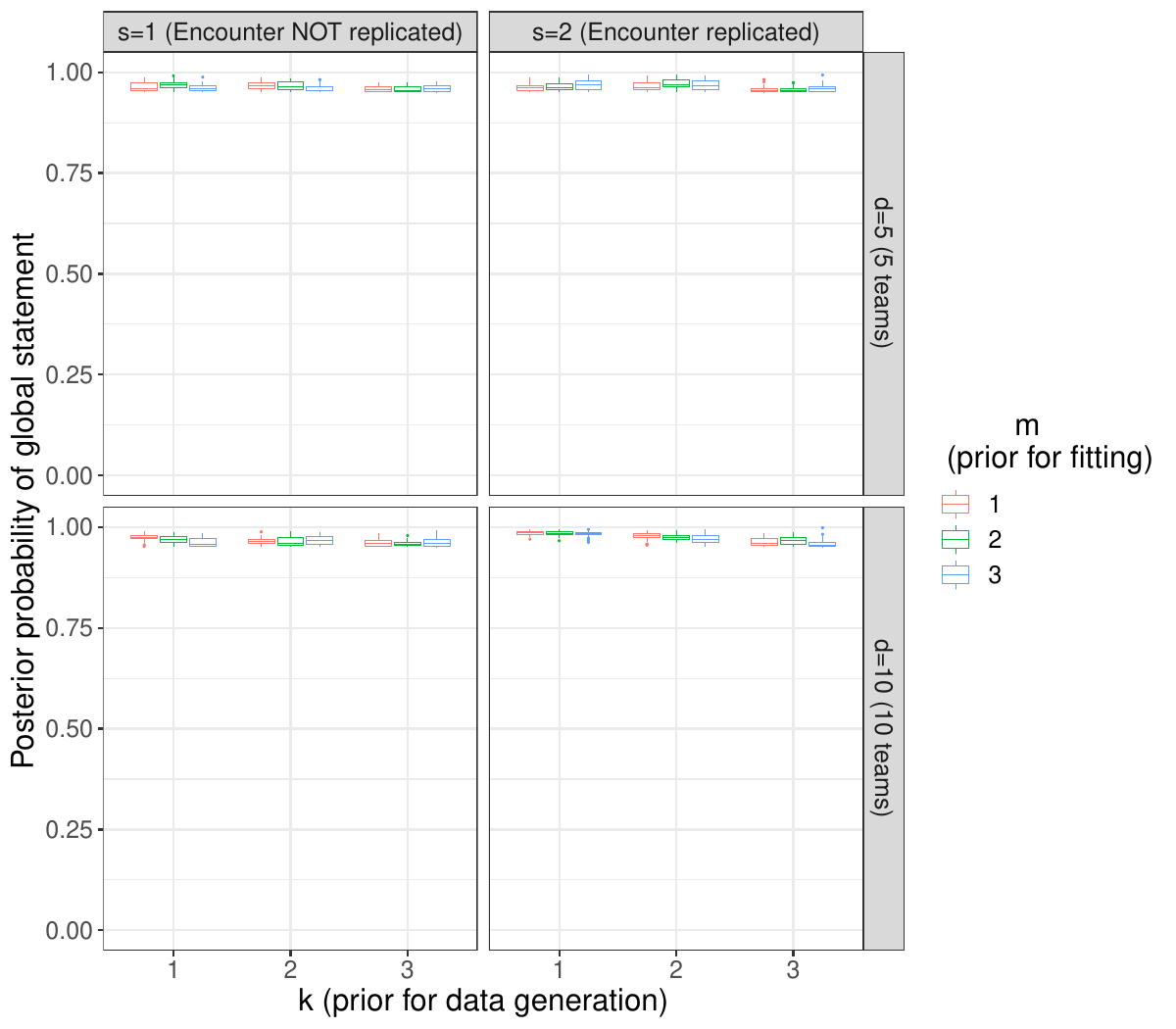}
\includegraphics[scale=0.45, trim={0 0cm 0cm 0cm},clip=true,page=2]{plots/FigurePlayers_JRSSC.pdf}
\includegraphics[scale=0.45, trim={0 0cm 4.5cm -1cm},clip=true,page=3]{plots/FigurePlayers_JRSSC.pdf}
\includegraphics[scale=0.45, trim={0 0cm 0cm -1cm},clip=true,page=4]{plots/FigurePlayers_JRSSC.pdf}
\caption{Simulation results for players assessing the effect of prior distribution used for fitting ($m$), prior distribution used for data generation ($k$), number of parameters ($d$), and doubling the sample size ($s$), with $\bv = (k,d,s,m)$.  
Top-left boxplot displays the posterior probability of the global statements;
top-right plot shows the frequentis coverage;
bottom-left boxplot presents the number of local statements within global statements;
bottom-right boxplot shows the distribution across global statements of the average number of lineups within local statements. The displayed plots summaries results over 50 synthetic datasets.}
\label{fig:simulation.resultsS1}
\end{figure}

\begin{figure}[H]
\centering
\includegraphics[scale=0.45, trim={0 0cm 4.5cm 0cm},clip=true,page=5]{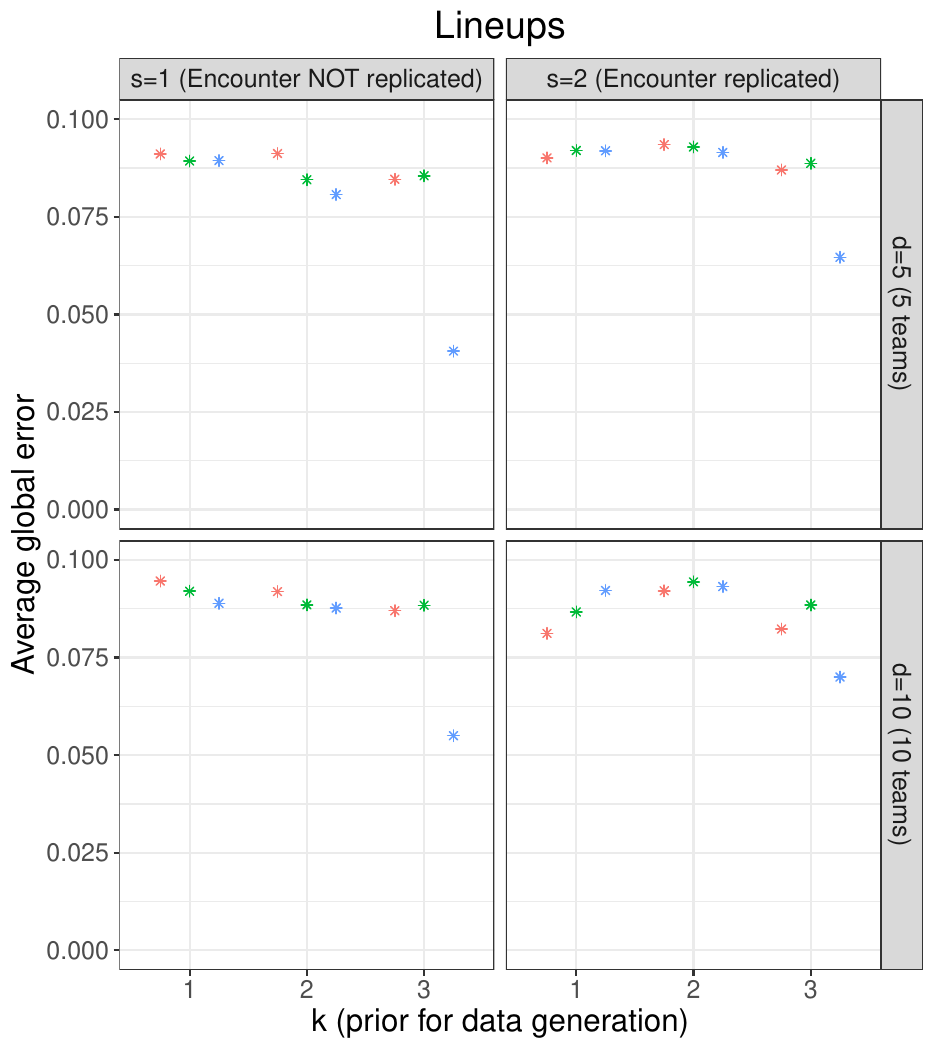}
\includegraphics[scale=0.45, trim={0 0cm 0cm 0cm},clip=true,page=6]{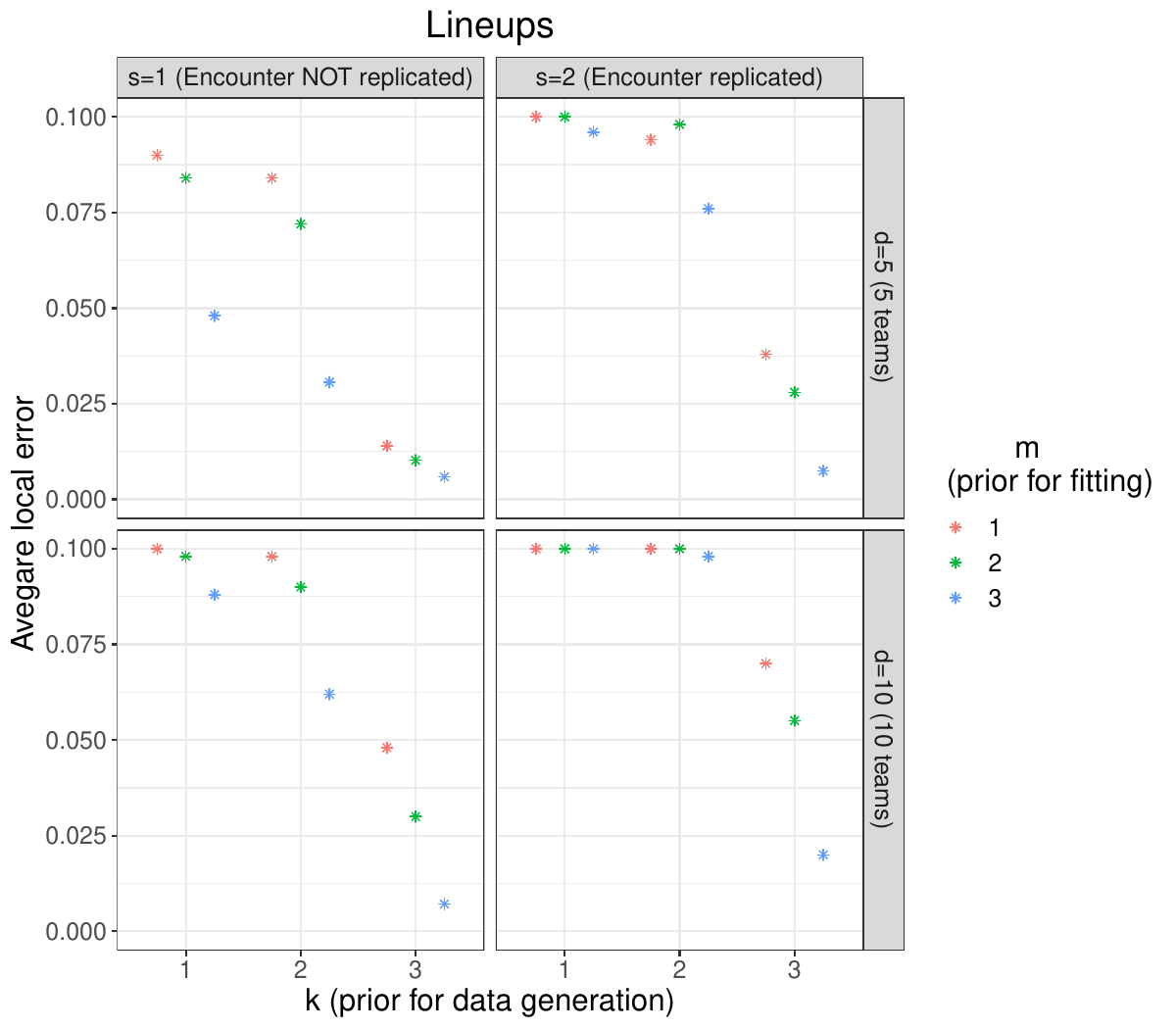}
\includegraphics[scale=0.45, trim={0 0cm 4.5cm -1cm},clip=true,page=5]{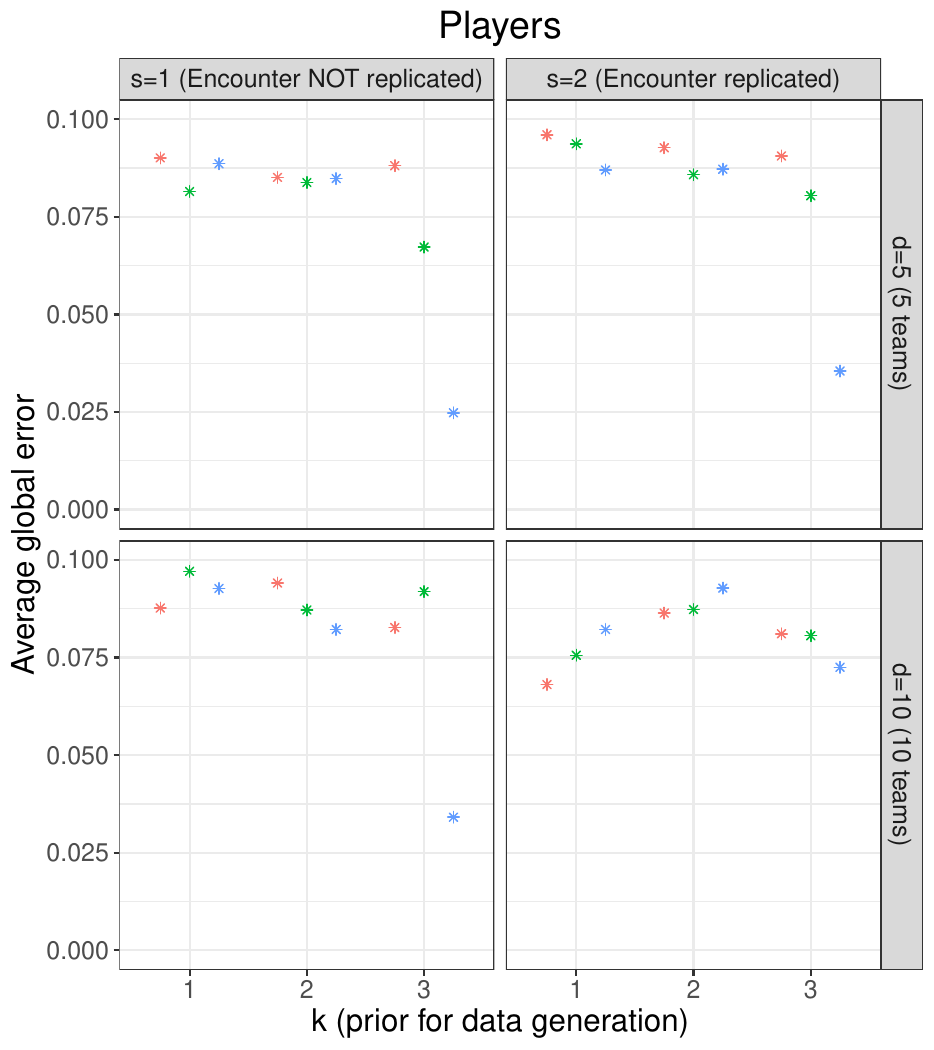}
\includegraphics[scale=0.45, trim={0 0cm 0cm -1cm},clip=true,page=6]{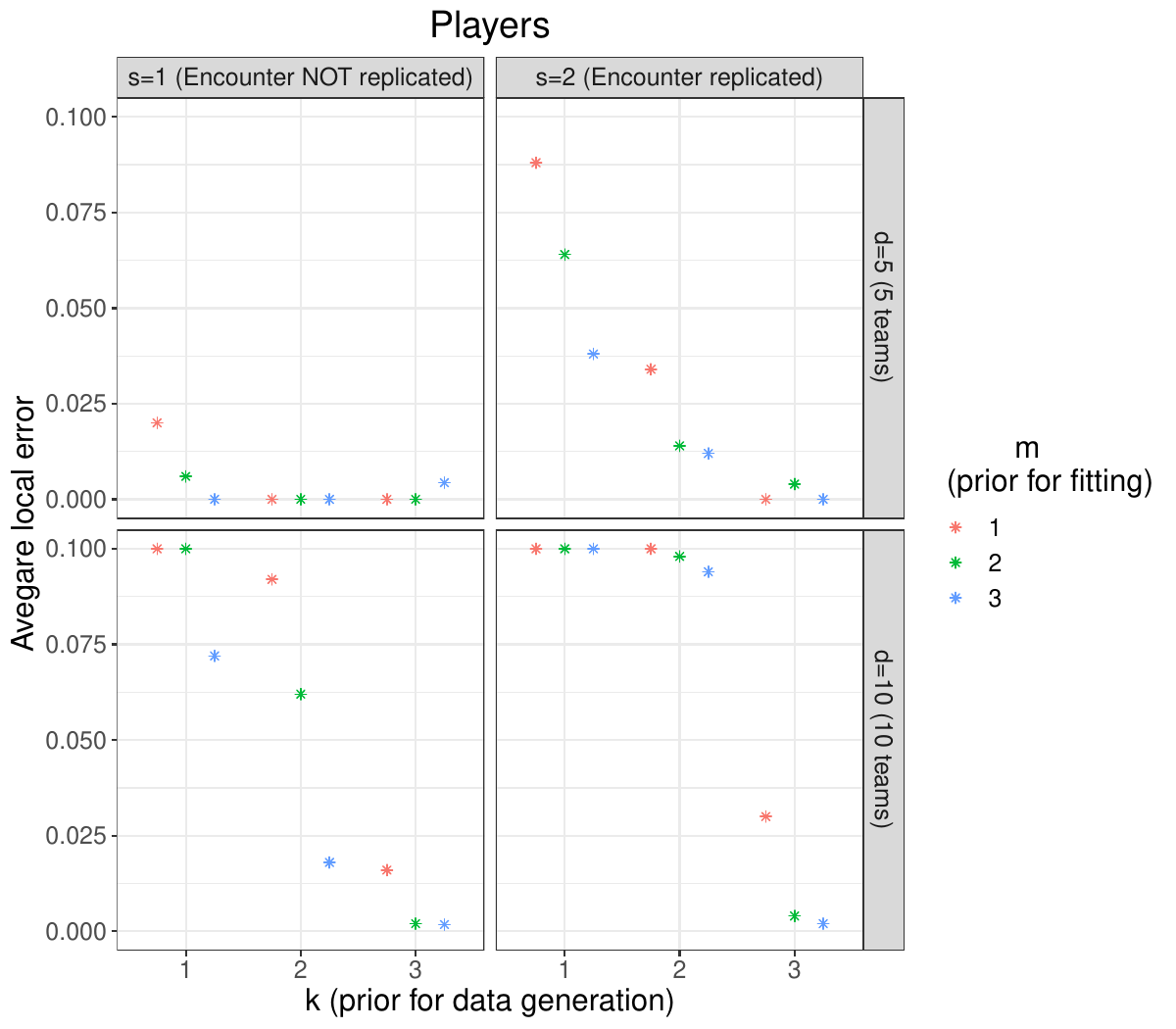}
\caption{Simulation results for the effect of prior distribution ($m$), similarity among ``true'' abilities ($k$), number of parameters ($d$), and doubling the sample size ($s$), with $\bv = (k,d,s,m)$.
Plots in first and second columns of panels correspond to the global and local errors of the statements across fifty synthetic datasets, respectively.}
\label{fig:simulation.resultsS2}
\end{figure}

Figure \ref{fig:simulation.results.CPCR} displays the results under the cumulative probability consensus ranking (CPCR) methodology proposed by \cite{vitelli2018probabilistic}. As observed for lineups (see Section \ref{sec:simulation.resuls} of the main manuscript), posterior probability and frequentist coverage for each CPCR are similar for players with values less than or around $0.3$ for top players. Figure \ref{fig:simulation.results.players.joint.CPCR} shows the effect that $L_0$ has on the probabilities for joint CPCRs for players defined as $\bigcap_{l=L_0}^L [\xi_{(1:l)}\mbox{ ranked $1$st, $2$nd, $\ldots$, or $l$th}]$. The larger $L_0$ (less players in the joint statement), the higher the probability. Similar to lineups, these probabilities quickly decay to zero as $L_0$ decreases. Notice that joint CPCRs concentrating a probability close to $0.9$ require $L_0 = 76$ for $d = 5$ ($5$ teams = $78$ players) and  $L_0 = 152$ for $d = 10$ ($10$ teams = $167$ players). For the case $L_0 = 152$ with $d = 5$, the most precise statement in $\bigcap_{l=152}^{167} [\xi_{(1:l)}\mbox{ ranked $1$, $2$, $\ldots$, or $l$}]$ states that player $\xi_{(1:152)}$ is ranked $1$st, $2$nd, $\ldots$, or $152$th out of $167$ players, which again does not convey too much information. Compared to our proposed global statements, the lack of precision could be a disadvantage when using joint CPCRs in scenarios with high uncertainty.

\begin{figure}
\centering
\includegraphics[scale=0.45, trim={0 0cm 4.5cm 0cm},clip=true,page=1]{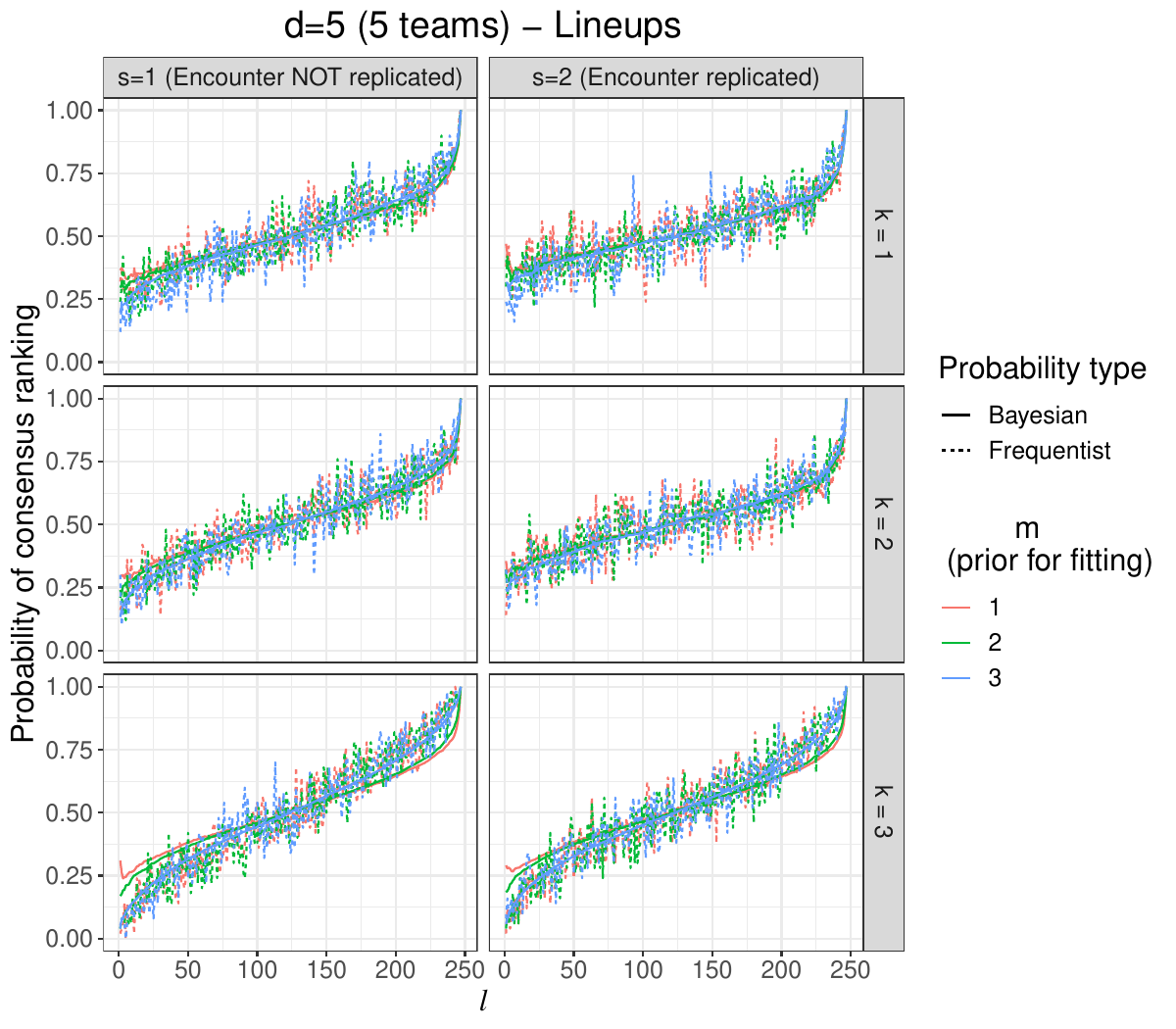}
\includegraphics[scale=0.45, trim={0 0cm 0cm 0cm},clip=true,page=2]{plots/FigureCPCR_Lineups_BA.pdf}
\includegraphics[scale=0.45, trim={0 0cm 4.5cm 0cm},clip=true,page=1]{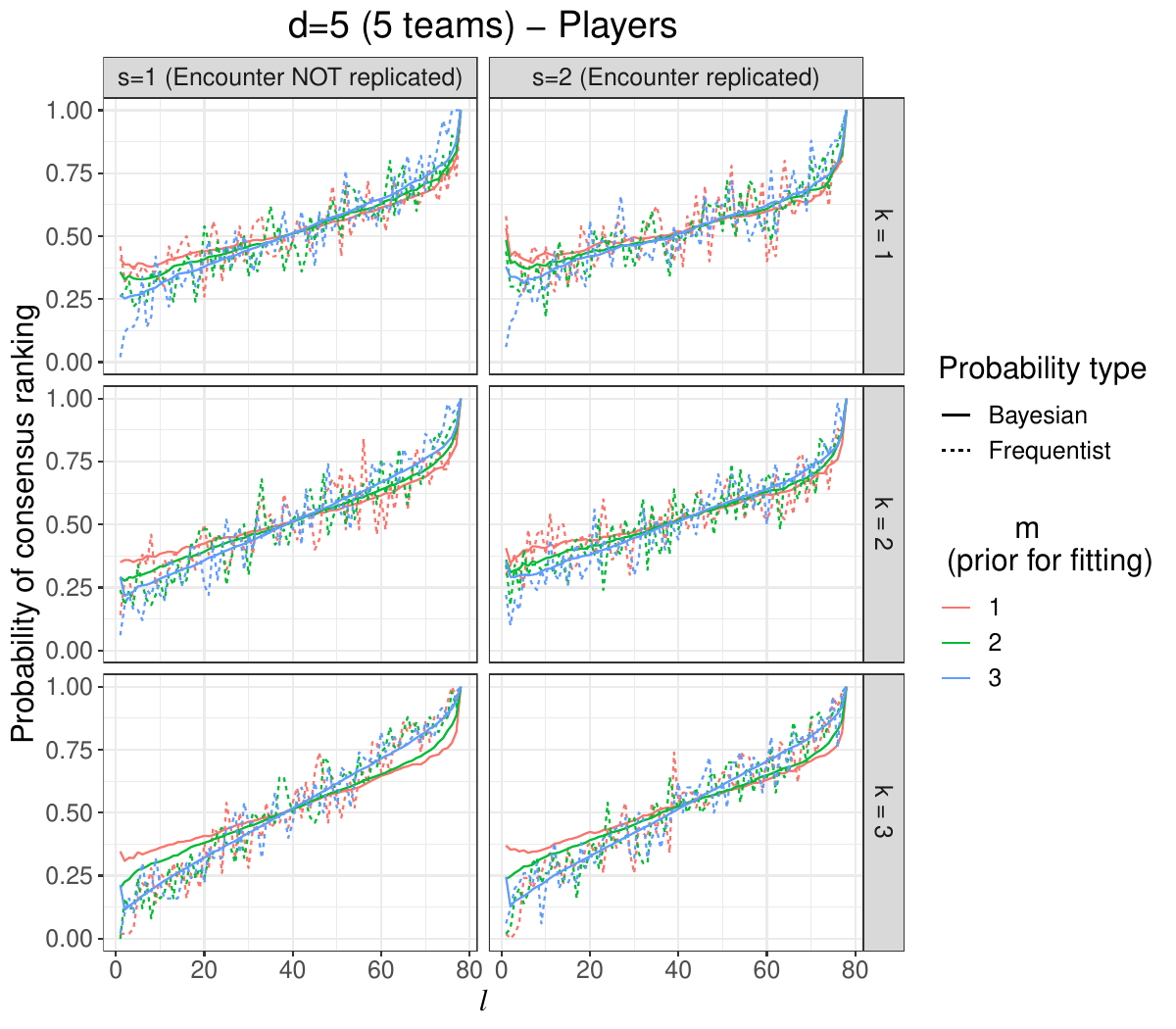}
\includegraphics[scale=0.45, trim={0 0cm 0cm 0cm},clip=true,page=2]{plots/FigureCPCR_Players_BA.pdf}
\caption{Simulation results for the CPCRs for lineups and players assessing the effect of prior distribution used for fitting (captured by $m$), prior distribution used for data generation (captured by $k$), number of parameters (captured by $d$), and doubling the sample size (captured by $s$) with $\bv = (k,d,s,m)$.  
Each plot displays the average posterior probability and frequentist coverage of each CPCR (i.e., $[\xi_{(1:l)}\mbox{ ranked $1$, $2$, $\ldots$, or $l$}]$ for players and $[\tau_{(1:l)}\mbox{ ranked $1$, $2$, $\ldots$, or $T$}]$ for lineups). These plots summarize results over 50 synthetic datasets.}
\label{fig:simulation.results.CPCR}
\end{figure}

\begin{figure}
\centering
\includegraphics[scale=0.45, trim={0 0cm 4.5cm 0cm},clip=true,page=3]{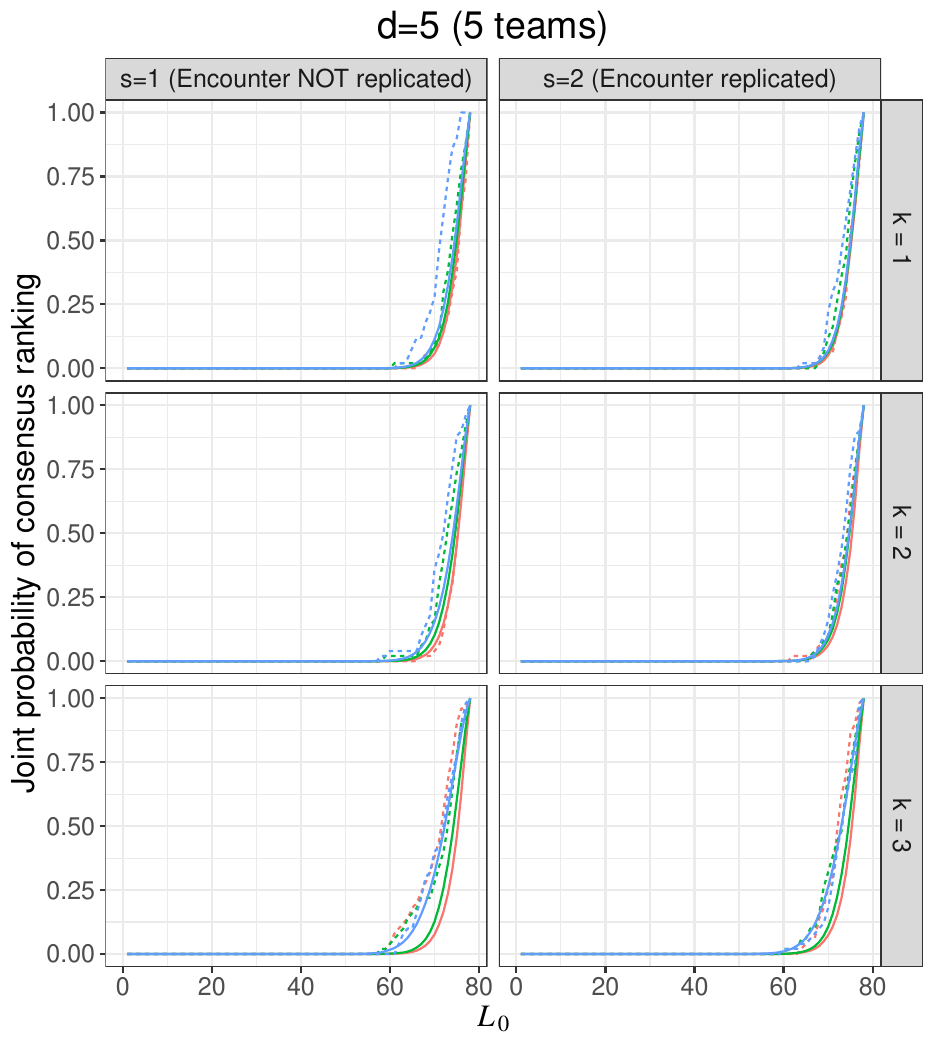}
\includegraphics[scale=0.45, trim={0 0cm 0cm 0cm},clip=true,page=4]{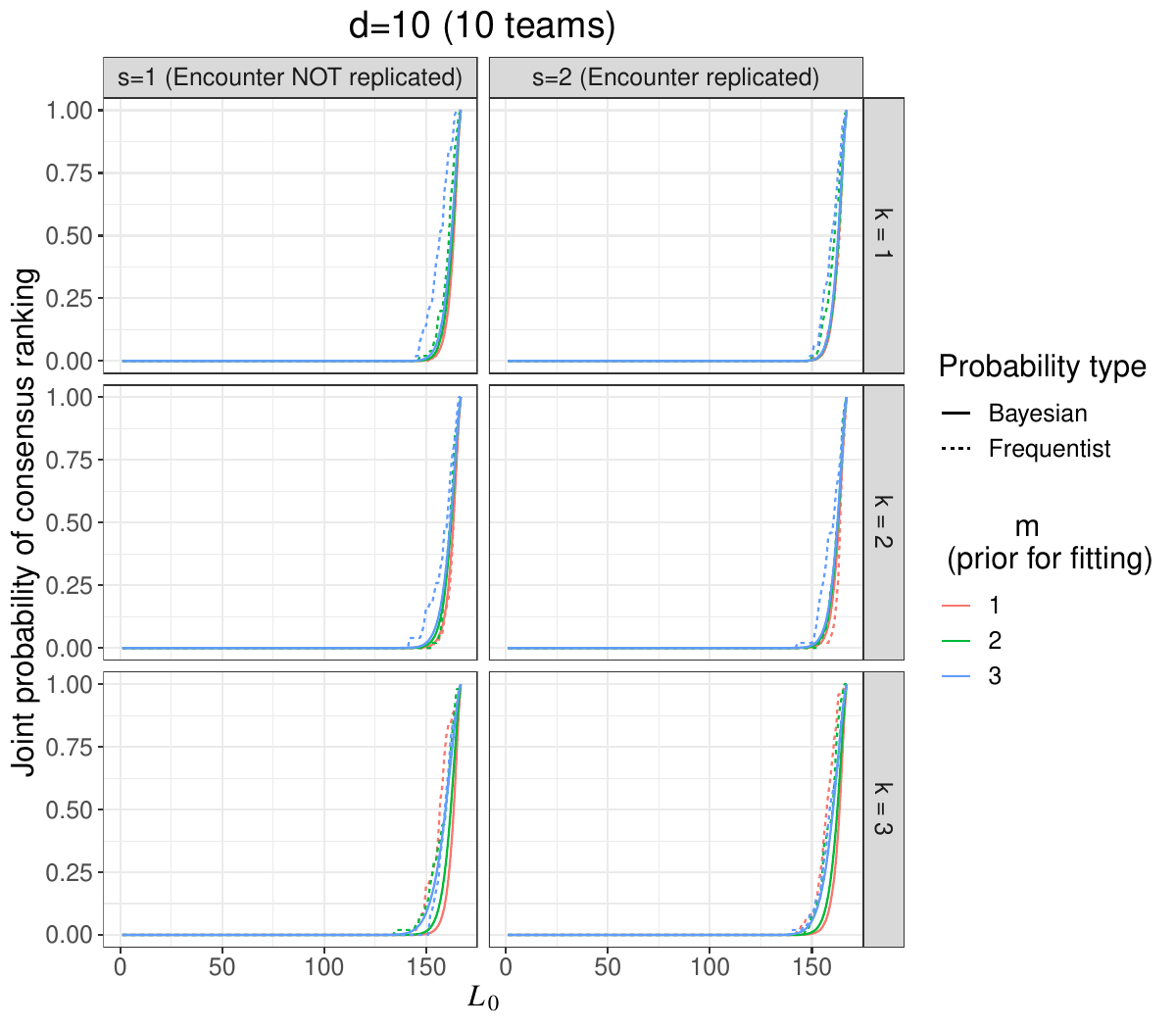}
\caption{  Simulation results for joint CPCRs for players assessing the effect of prior distribution used for fitting (captured by $m$), prior distribution used for data generation (captured by $k$), number of parameters (captured by $d$), and doubling the sample size (captured by $s$) with $\bv = (k,d,s,m)$.  Each plot displays the average posterior probability and frequentist coverage of the joint CPCRs (i.e., $\bigcap_{l=L_0}^L [\xi_{(1:l)}\mbox{ ranked $1$, $2$, $\ldots$, or $l$}]$). These plots summarize results over 50 synthetic datasets.}
\label{fig:simulation.results.players.joint.CPCR}
\end{figure}

Lastly, in our simulations we also considered the case of encounters being comprised of lineups that didn't appear in the NBA data.  That is, lineups that were not observed.  Ordering hypothetical lineups is clearly of interest to general managers and coaches as they work through the process of team building.  To generate encounters that involved unobserved lineups we focused on two teams; Cleveland and Minnesota. For each  simulated dataset, we generated 50 unobserved lineups for each team (note that the ``true'' ability of each unobserved lineup is still known).  We then selected the 10 best lineups from Cleveland (based on the total ability of the five players) and 10 lineups from Minnesota with the lowest ability.  Using these twenty lineups we run our ordering procedure comparing lineups within Minnesota and Cleveland and also between the two teams.  Thus, for each synthetic dataset, we obtain a global statement based on the 20 unobserved lineups and compute its posterior probability.  As can be seen in the top left plot of Figure \ref{fig:simulation.results.unobserved.lineups} the global statements were all accompanied with large posterior probability.  The other three plots in Figure \ref{fig:simulation.results.unobserved.lineups} mirror results associated with the simulation that was described in the main document.  Mainly, that increasing the number of encounters and teams increases the density of the local and global statements regardless of the prior employed to estimate player abilities.  And priors that are more diffuse with regards to player abilities produce global and local statements that are more sparse, but also more accurate.
\begin{figure}[H]
\centering
\includegraphics[scale=0.45, trim={0 0cm 4.5cm 0cm},clip=true,page=1]{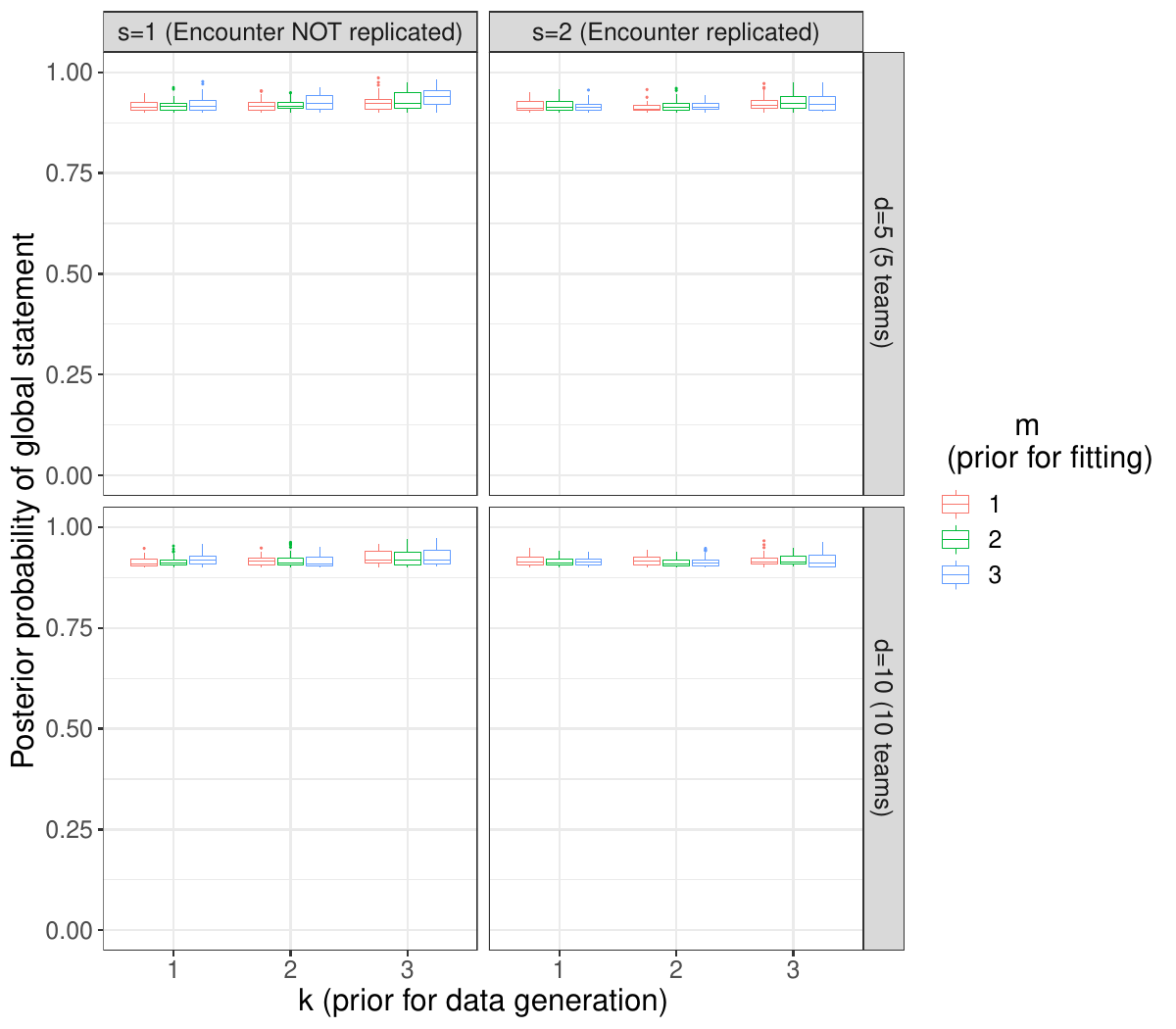}
\includegraphics[scale=0.45, trim={0 0cm 0cm 0cm},clip=true,page=2]{plots/FigureUnobservedLineups_BA.pdf}
\includegraphics[scale=0.45, trim={0 0cm 4.5cm -1cm},clip=true,page=3]{plots/FigureUnobservedLineups_BA.pdf}
\includegraphics[scale=0.45, trim={0 0cm 0cm -1cm},clip=true,page=4]{plots/FigureUnobservedLineups_BA.pdf}
\caption{ Simulation results for unobserved lineups within Minnesota and Cleveland. The displayed results assess the effect of prior distribution used for fitting ($m$), prior distribution used for data generation ($k$), number of parameters ($d$), and doubling the sample size ($s$), with $\bv = (k,d,s,m)$.  
Top-left boxplot displays the posterior probability of the global statements;
top-right plot shows the frequentis coverage;
bottom-left boxplot presents the number of local statements within global statements;
bottom-right boxplot shows the distribution across global statements of the average number of lineups within local statements. The displayed plots summaries results for $10$ unobserved lineups for each and over 50 synthetic datasets.}
\label{fig:simulation.results.unobserved.lineups}
\end{figure}

\section{Additional results for player comparisons application} \label{supp.player.comparison}

We applied the proposed approach to find global statements for the player abilities $\xi_{l}$s. Figure \ref{fig:all_players_teamsS1} shows the caterpillar plot associated with the posterior interquartile range for each of the $\xi_{l}$ given the NBA lineup data. Notice the large amount of overlap in player abilities making any ranking highly uncertain.
Jamal Crawford of the Atlanta Hawks was selected as the reference player (he came first after alphabetizing team names and player names).  Thus, these ability interquartile ranges are relative to Crawford's.  
\begin{figure}[H]
\centering
\includegraphics[width=0.7\textwidth, trim={0 0cm 0cm 2cm},clip=true]{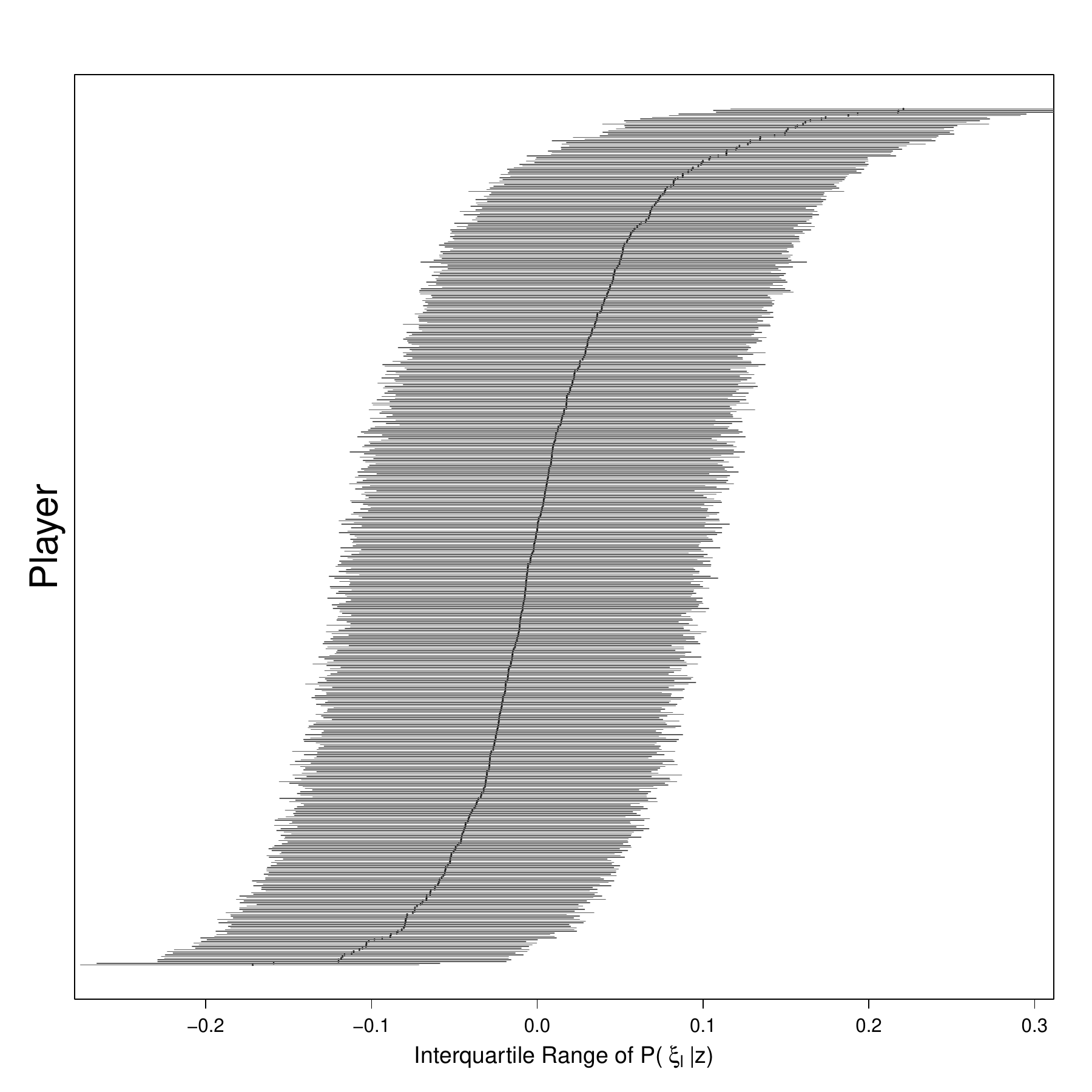}
\caption{Graphical depiction of interquartile range associated with each player's ability's posterior distribution.}
\label{fig:all_players_teamsS1}
\end{figure}

Figure \ref{fig:all_players_teamsS2} displays the global statements for all 510 players.  The left plot depicts results for $t=q=0.1$.  The posterior probability of this global statement is 0.9, but the statement is very sparse.  In fact, there are only a handful of players whose local statement is not empty and little can be said if this amount of precision is needed.   The right plot corresponds to the player global state for $t=q=0.25$ and has posterior probability 0.999.  Taking on the added imprecision does result in a more sparse global statement, but about half of the 510 players still have empty local statements.  Since the spread of each player's posterior distribution is large, each additional player that is included in the global statement introduces quite a bit of uncertainty.  As a result global statements that include all 510 players will be sparse regardless of the values ultimately used for $t$, $q$, $\alpha$, and $\gamma$.
\begin{figure}[H]
\centering
\includegraphics[scale=0.4, trim={0cm 0cm 0cm 0cm},clip=true]{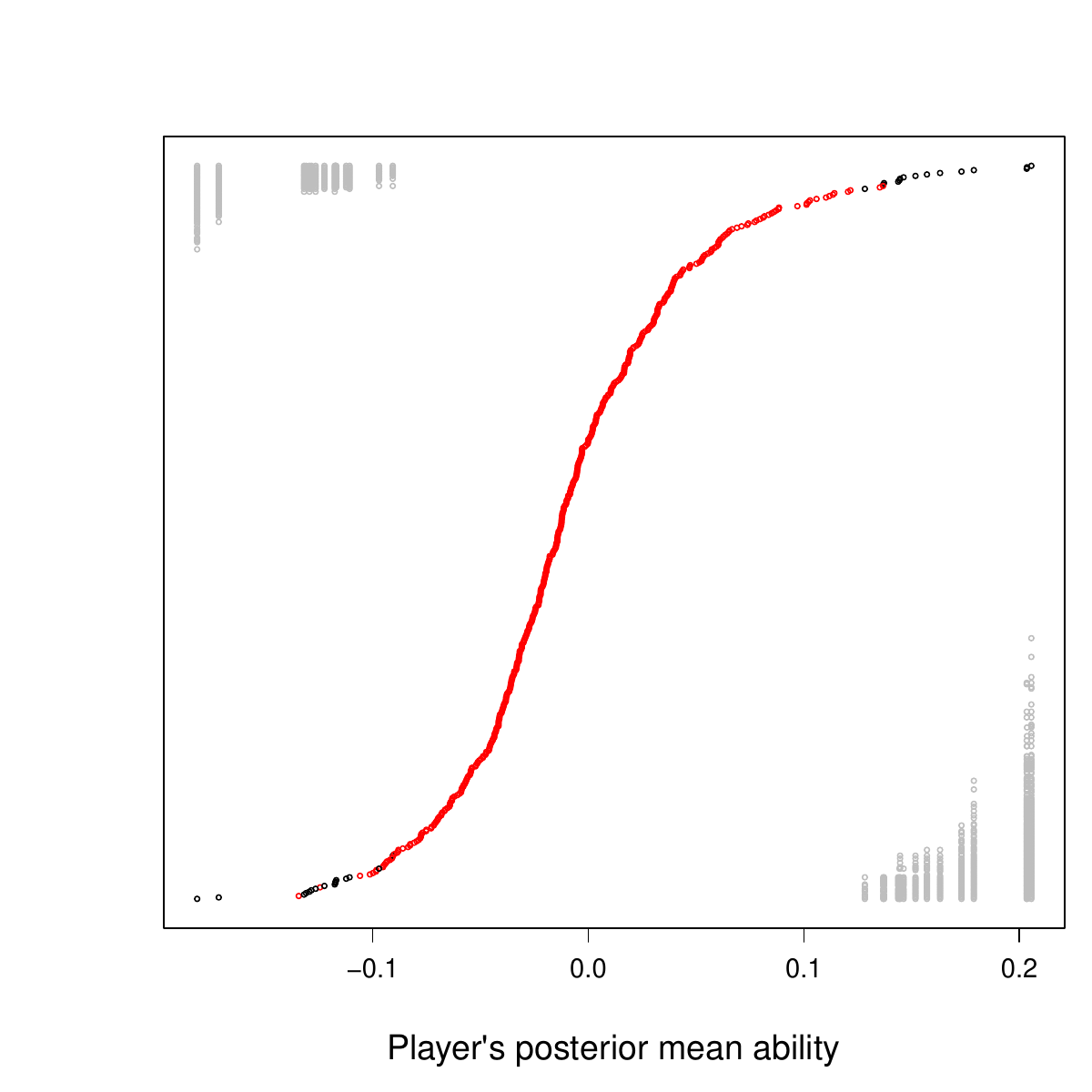}
\includegraphics[scale=0.4, trim={0cm 0cm 0cm 0cm},clip=true]{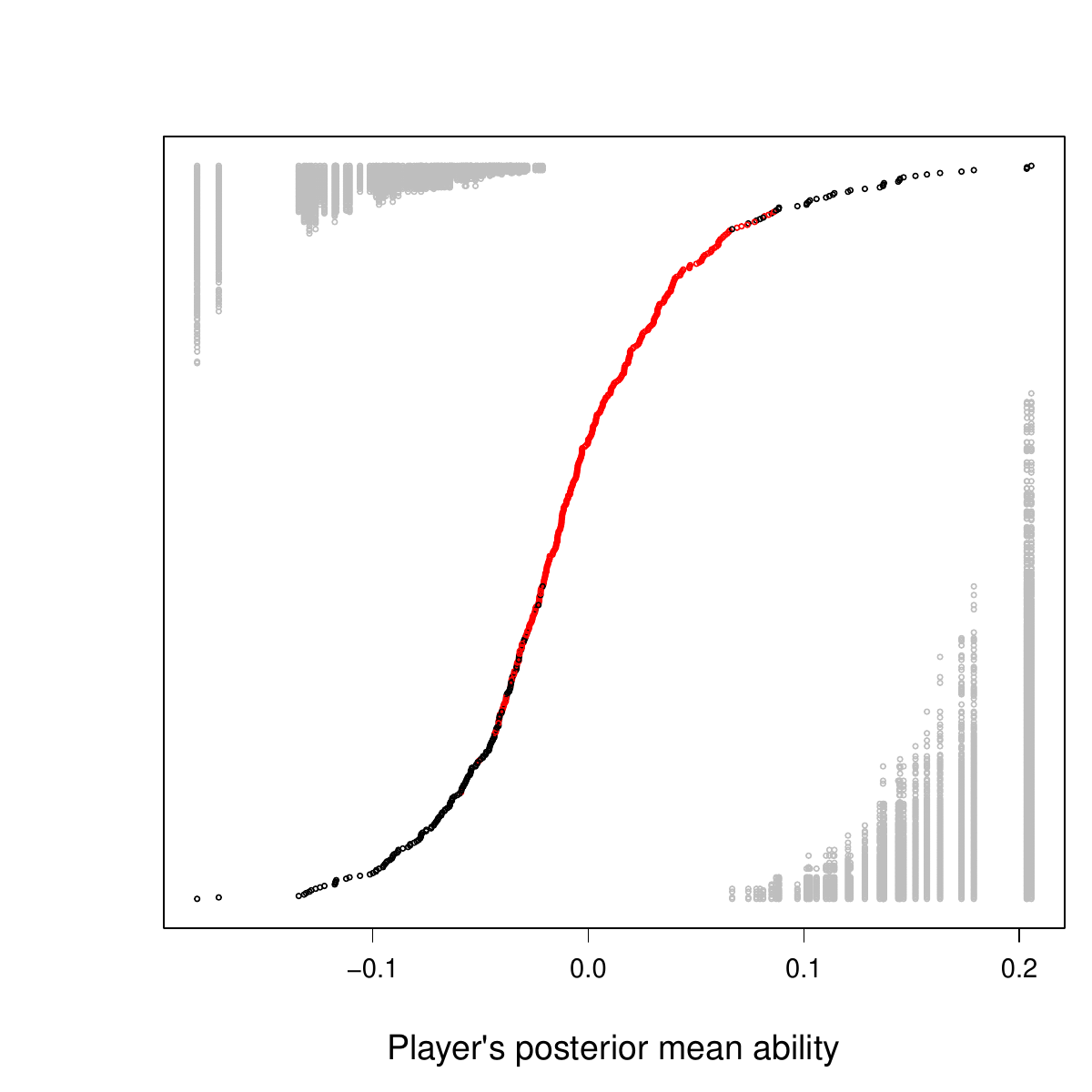}
\caption{Graphical depiction of global statement based on the 510 players that participated in an NBA game during the 2009-2010 season.  The left plot corresponds to $t = q = 0.1$ while the right to $t=q=0.25$ }
\label{fig:all_players_teamsS2}
\end{figure}

\end{document}